\documentclass[journal,10pt,twocolumn,twoside]{IEEEtran}

\IEEEoverridecommandlockouts
\usepackage{amsmath}
\usepackage{amsfonts}
\usepackage{cases}
\usepackage{setspace}
\usepackage{fancybox}
\usepackage{subfigure}
\usepackage{epsfig}
\usepackage{graphicx}
\usepackage{epstopdf}
\usepackage{float}
\usepackage{multirow}
\usepackage{color}
\usepackage{amsmath}
\usepackage{multirow}
\usepackage{indentfirst}
\usepackage{dsfont}
\usepackage{amsfonts}
\usepackage{times,amsmath,color,amssymb,epsfig,cite,subfigure,algorithm,algorithmic}

\newenvironment{proof}[1][Proof]{\begin{trivlist}
\item[\hskip \labelsep {\bfseries #1}]}{\end{trivlist}}

\newcommand{\qed}{\nobreak \ifvmode \relax \else
      \ifdim\lastskip<1.5em \hskip-\lastskip
      \hskip1.5em plus0em minus0.5em \fi \nobreak
      \vrule height0.40em width0.6em depth0.25em\fi}
\newcommand{\tabincell}[2]{\begin{tabular}{@{}#1@{}}#2\end{tabular}}
\newtheorem{lemma}{Lemma}
\newtheorem{theorem}{Theorem}

\begin{document}

\title{On the Design of Constant Modulus Probing Waveforms with Good Correlation Properties for MIMO Radar via Consensus-ADMM Approach}

\author{Jiangtao~Wang,
        Yongchao~Wang,~\IEEEmembership{Member,~IEEE}
\thanks{Manuscript received November 06, 2018; revised April 15, 2019 and
June 21, 2019; accepted June 25, 2019. Date of publication XX xx,
2019; date of current version July 8, 2019. The associate editor coordinating
the review of this manuscript and approving it for publication was Dr. Ian Clarkson.
This work was supported in part by National Science Foundation of China
under grant 61771356 and the Fundamental Research Funds for the Central Universities.
This work was presented in part at 2018 International Conference on Acoustics, Speech and Signal Processing (ICASSP) in Calgary, Alberta, Canada.
 (Corresponding author: Yongchao Wang.)}
\thanks{J. Wang and Y. Wang are
State Key Lab. of ISN, Xidian University, Xi'An, China, 710071,  (e-mail: jt.wang@stu.xidian.edu.cn, ychwang@mail.xidian.edu.cn).}}

\markboth{}%
{}

\maketitle

\begin{abstract}
In this paper, we design constant modulus probing waveforms with good correlation properties for collocated multi-input multi-output (MIMO) radar systems. The main content is as follows:
first, we formulate the design problem as a fourth order polynomial minimization problem with constant modulus constraints. Then,
by exploiting introduced auxiliary variables and their inherent structures, the polynomial optimization model is equivalent to a non-convex consensus minimization problem.
Second, a customized alternating direction method of multipliers (ADMM) algorithm is proposed to solve the non-convex problem approximately. In the algorithm, all the subproblems can be solved analytically. Moreover, all subproblems except one subproblem can be performed in parallel.
Third, we prove that the customized ADMM algorithm is theoretically-guaranteed convergent if proper parameters are chosen.
Fourth, two variant ADMM algorithms, based on stochastic block coordinate descent and accelerated gradient descent, are proposed to reduce computational complexity and speed up the convergence rate.
Numerical examples show the effectiveness of the proposed consensus-ADMM algorithm and its variants.
\end{abstract}


\IEEEpeerreviewmaketitle

\section{Introduction}
\IEEEPARstart{M}{\MakeLowercase u}ltiple-input multiple-output (MIMO) radar system is regarded as a promising paradigm for the next generation radar systems.
Unlike the standard phased-array radar to transmit scaled versions of a single waveform, probing signals, transmitted via different antennas in the MIMO radar system,  are independent. Through this additional waveform diversity, MIMO radar owns superior capabilities compared with the traditional phased-array radar, such as higher spatial resolution, more flexible beampattern, and better detection performance \cite{Li_09}\cite{Li_07}.
MIMO radar system can be classified into two categories: distributed and collocated.
In the former, transmitters are widely separated in space and each of them can provide an independent view of the target, which can improve detection performance \cite{Haimovich_08}\cite{He_10}.
In the latter, antennas in the transmitter are placed in close proximity and different probing signals
from various collocated antennas can generate various desired
beampatterns, leading to an improved directional resolution and interference
rejection capability \cite{Lehmann_06}--\hspace{-0.001cm}\cite{Chen_08}.

Probing signal waveforms play a central role in the signal processing
performance of a MIMO radar system.
Specifically, since matching the desired spacial beampatterns and lowering spacial correlations levels can increase spacial directional gain and eliminate clutter interference from other directions, a lot of researchers have been attracted to designing probing signal waveforms to meet these goals in recent years.
Authors in \cite{Stoica_07_probing} and \cite{Li_08} matched the waveform covariance matrix to the desired beampattern through a semidefinite programming method, then exploited the cyclic algorithm to synthesize the constant modulus waveform and pursued good auto-/cross-correlation properties. In \cite{Wang_12}, authors formulated the waveform design problem as a fourth order polynomial minimization problem with constant modulus constraints, then proposed a quasi-Newton solving algorithm to approximate the model's optimal solution. Moreover, the approach can be applied to the scenario of desired low correlation sidelobe levels within certain lag intervals. The authors in \cite{Hua_13} focused on the direct or indirect control of mainlobe ripples in the beampattern design problem. They reformulated the design as a feasibility problem with the lowest system cost. To achieve a high signal to interference plus noise ratio and low sidelobe levels performance, a fixed waveform covariance matrix was proposed in \cite{Ahmed_2014}. However, the matrix does not exploit the full waveform diversity. In \cite{Zhang_15}, the authors proposed a novel transmit beampattern matching design one-step method, which obtains the transmit signal  matrix by unconstrained optimization. The drawback of the waveforms generated by this method is that their envelope is not constant modulus. To reduce the computational complexity, a closed-form covariance matrix design method was proposed in \cite{Lipor_14} based on discrete Fourier transform (DFT). The authors in \cite{Bouchoucha_15} and \cite{Bouchoucha_17} also applied the DFT-based technique to a planar-antenna-array, and developed a finite-alphabet constant-envelope waveforms design algorithm for the desired beampattern. However, the performance of the DFT-based method is slightly worse for a small number of antennas.
 The authors in \cite{Aubry_16} studied the robust transmit beampattern design problem and exploited the semidefinite relaxation technique to treat non-convex optimization problems. In \cite{Aldayel_17_conf} \cite{Aldayel_17}, the authors exploited  successive convex relaxation techniques to handle non-convex quadratic equality constraints in the constant modulus waveform design problem. In \cite{Cheng_17}, the authors proposed a double cyclic alternating direction method of multipliers (D-ADMM) algorithm to solve the non-convex beampattern design problem and in \cite{Cheng_18}, they considered the joint optimization problem of the covariance matrix and antenna position. In \cite{Zhao_18}, the authors applied the majorization-minimization technique to match the desired transmit beampattern, which enjoys faster convergence than D-ADMM.
The authors in \cite{Yu_18} focused on MIMO radar waveform design under the constant modulus and similarity constraints. They proposed a sequential iterative algorithm based on the block coordinate descent (BCD) framework, which has shown its superiority compared with the CA approach in \cite{Li_08}.
In \cite{Yu_19}, the authors considered the constant modulus waveform design to achieve a desired wideband MIMO radar beampattern with space-frequency nulling.
In each algorithm iteration, the authors optimized the original non-convex problem's approximation version meaning that the proposed algorithm can be executed in parallel.
Besides the above works, some researchers synthesized transmit waveforms under some practical constraints, such as mainlobe ripple constraints \cite{Wen_18},
spectral shape constraints \cite{Liang_15}, constant modulus constraints \cite{Liang_16}, similarity constraints \cite{Cheng_14} \cite{Liang_18}, and transmitted power constraints \cite{Li_17}. However, these works only focus on the synthesized beampattern design problem and pay little attention to the correlation properties of the waveforms.

 In this paper, we extend our previous work in \cite{ADMM_con} and propose a consensus-ADMM approach to design constant modulus probing waveforms, which can match the desired spacial beampatterns while suppressing the spacial auto-correlation and cross-correlation levels in the collocated MIMO radar system. Its main contributions are as follows.

\begin{itemize}
\item \emph{Consensus problem formulation}:
  the design problem is formulated as a fourth order polynomial minimization problem with constant modulus constraints. Then, by introducing  auxiliary variables, it is further equivalent to a non-convex consensus minimization problem.
\item \emph{Parallel solving algorithm}: consensus-ADMM is customized to solve the non-convex consensus problem approximately. In the implementation, all the subproblems can be solved analytically. Moreover, except one subproblem, all subproblems can be performed in parallel. This  favourable execution architecture is the main advantage of the proposed consensus-ADMM over state-of-the-art techniques, which is very suitable for practical implementation.
\item \emph{Theoretically-guaranteed performance}: we prove that the solving algorithm is guaranteed convergent to a stationary point of the non-convex optimization problem if proper parameters are chosen.
\item \emph{Improvement strategies}: two variant ADMM algorithms, based on stochastic block coordinate descent (SBCD) and accelerated gradient descent (AGD), are proposed to reduce computational complexity and speed up the convergence rate.
\end{itemize}

The rest of the paper is organized as follows. In Section II, we formulate the beampattern design problem to a non-convex consensus minimization problem. In Section III, consensus-ADMM is customized to solve the non-convex minimization problem. The performance analysis, including convergence and computational complexity of the proposed consensus-ADMM algorithm, are presented in Section IV. Two variant algorithms, named by consensus-ADMM-SBCD and consensus-ADMM-AGD, are given to improve computational complexity and convergence performance of the solving algorithm  respectively in Section V. Finally, Section VI demonstrates the effectiveness of the proposed consensus-ADMM algorithms and the conclusions are given in Section VII.

\emph{Notation}: bold lowercase and uppercase letters denote column vectors and matrices and italics denote scalars. $\mathds{R}$ and $\mathds{C}$ denote the real field and the complex field respectively.
The superscripts $(\cdot)^*$, $(\cdot)^T$ and $(\cdot)^H$ denote  conjugate operator, transpose operator and conjugate transpose operator respectively.
$x_i$  denotes the $i$-th element of vector $\mathbf{x}$.
$|\cdot|$ denotes the absolute value. The subscript $\|\cdot\|_2$ denotes the Euclidean vector norm and $\|\cdot\|_F$ denotes the Frobenius matrix norm.
$\nabla(\cdot)$ represents the gradient of a function. $\rm{Re}(\cdot)$ takes the real part of the complex variable and ${\rm Tr}(\cdot)$ denotes the trace of a matrix.
$\langle \cdot,\cdot\rangle$ and $\otimes$ are the dot product operator and convolution operator respectively.
${\rm vec}(\cdot)$ vectorizes a matrix by stacking its columns on top of one another and ${\rm mat}(\cdot,N,M)$ reshapes a vector to an $N\times M$ matrix. $\Pi(\cdot)$ denotes the projection operator. $\mathbb{E}[\cdot]$ performs the expectation of the variables and ${\mathbf{I}}$ denotes an identity matrix.

\section{System Model And Problem Formulation}
\label{sec:format}
\subsection{System Model}
\begin{figure}[htb]
  \centering
\centerline{\psfig{figure=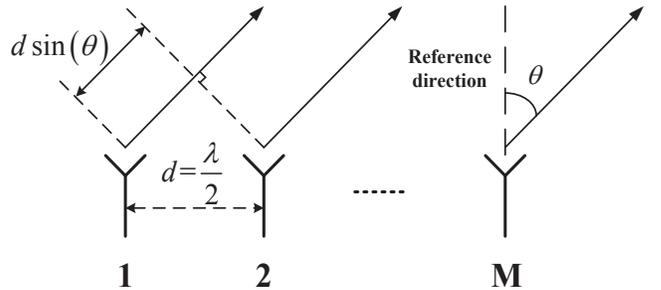,width=8.5cm,height=4.0cm}}
\caption{MIMO radar transmitter equipped with $M$ antennas.}
\end{figure}
Consider a MIMO radar system equipped with $M$ antennas in a uniform linear array as shown in Figure 1.
In the system, we set the inter-element spacing $d = \frac{\lambda}{2}$, where $\lambda$ is the signal wavelength. The spacial direction $\theta$ belongs to angle set $\Theta=(-90^\circ,90^\circ)$, which represents the antenna scanning scope.
The steering vector $\mathbf{a}\in\mathds{C}^M$ at direction $\theta$ is given by
\begin{equation}\label{definiton steer vector}
\mathbf{a}_\theta = \left[1,e^{j\pi \sin\theta},\cdots,e^{j\pi(M-1) \sin\theta}\right]^T.
\end{equation}
The probing waveform transmitted by the $m$-th antenna is denoted by $\mathbf{x}_m=[x_{1,m}, \cdots ,x_{N,m}]^T,m=1,\cdots,M$.
Then, the waveforms transmitted by the MIMO radar system can be expressed by the following $N$-by-$M$ matrix
\begin{equation}\label{definintion X}
\mathbf{X}=[\mathbf{x}_1, \cdots ,\mathbf{x}_M].
\end{equation}
The synthesized signal at direction $\theta$ (far field) is
\begin{equation}\label{synthesis signal}
 \mathbf{s}_\theta=\mathbf{X}\mathbf{a}_\theta.
\end{equation}
The beampattern, which describes the power distribution at direction $\theta$, is defined as
\begin{equation}\label{definintion beampattern}
P_{\theta} = \mathbf{a}_\theta^{H}\mathbf{X}^H\mathbf{X}\mathbf{a}_\theta.
\end{equation}
To describe the correlation properties of the probing waveforms at time slot $n$, we define a $N$-by-$N$ off-line diagonal matrix $\mathbf{S}_n$ as follows
\[
\begin{split}
 &\qquad\qquad \quad n~\text{\rm zeros}\\
&{\bf{S}}_n=
\left[
  \begin{array}{cccccc}
    ~ & \overbrace{0~ \cdots ~0} & 1 & ~ & \mathbf{\scalebox{1.5}0} & ~ \\
    ~ & ~ & ~ &  \ddots &~ & ~  \\
    ~ & ~ & ~ & ~ & 1 & ~ \\
    ~ & \mathbf{\scalebox{2.5}0} & ~ & ~ & ~ & ~  \\
  \end{array}
\right].
\end{split}
\]
Through $\mathbf{S}_n$, the time-delayed signal can be expressed by $\mathbf{S}_n\mathbf{X}\mathbf{a}_\theta$.
Then, the spacial correlation of the probing waveforms and its delayed version can be obtained by
\begin{equation}\label{definition correlation}
P_{\theta_i,\theta_j,n} = \mathbf{a}_{\theta_i}^{H}\mathbf{X}^H\mathbf{S}_n\mathbf{X}\mathbf{a}_{\theta_j},
\end{equation}
where $\theta_i,\theta_j\in \hat{\Theta}\subset\Theta$ and $\hat{\Theta}=\{\theta_1,\cdots,\theta_K\}$ is the considered angle set of spacial directions.
Specifically, when $\theta_i=\theta_j$, $P_{\theta_i,\theta_i,n}$ denotes the spacial auto-correlation, otherwise $P_{\theta_i,\theta_j,n}$ means the spacial cross-correlation.

\subsection{Problem Formulation}
We optimize MIMO radar probing waveforms based on the following considerations: first, as mentioned in \eqref{definintion beampattern}, since the beampattern describes the spacial power distribution, we desire that it can match the directions of interest, which can decrease clutter components and extend the probing distance; second, since low auto-correlation sidelobes can increase spacial resolution and low cross-correlation levels can reduce  interferences from other directions, we desire that the optimized probing waveforms have low auto-correlation sidelobes and low cross-correlation levels; third,  in order to maximize the efficiency of the power amplifier in the MIMO radar transmitter, the probing waveforms  should be constant modulus,  i.e., $|x_{i,m}|=1,\ i = 1, \cdots ,N,\ m = 1, \cdots M$.

Based on the above considerations, we formulate the following optimization model to design MIMO radar probing waveforms
\begin{equation}\label{problem_model}
\begin{split}
&\hspace{0.5cm}\min_{\alpha,\mathbf{X}} \hspace{0.5cm} e(\alpha,\mathbf{X}) + P_c(\mathbf{X}),\\
&{\rm  subject\ to} \ |x_{i,m}|=1,\ i=1,\cdots,N,\ m=1,\cdots,M, \\
& \hspace{1.7cm}  \alpha\in(0,\alpha_{\rm max}],
\end{split}
\end{equation}
where
\begin{subequations}
\begin{equation}\label{problem e}
\begin{split}
&\hspace{-1.3cm}e(\alpha,\mathbf{X}) = \sum_{\theta\in \Theta}|\alpha \bar{P}_{\theta} - \mathbf{a}_\theta^{H}\mathbf{X}^H\mathbf{X}\mathbf{a}_\theta|^2, \\
\end{split}
\end{equation}
\begin{equation}\label{problem p_c}
\begin{split}
 \hspace{-0.3cm}P_c(\mathbf{X})\!\!=\!\!\!\!\!\!\sum_{n\in \mathcal{T}\backslash0}\sum_{\theta_i \in \hat{\Theta}}\!\!w_{\rm ac}^2|P_{\theta_i,\theta_i,n}|^2
 \!\!+\!\!\sum_{n\in \mathcal{T}}\!\!\!\! \sum\limits_{{\theta_i}\neq{\theta_j}\atop\theta_i,\theta_j \in \hat{\Theta}} \!\!\!\! w_{\rm cc}^2|P_{\theta_i,\theta_j,n}|^2,
\end{split}
\end{equation}
\end{subequations}
and $w_{\rm ac}$ and $w_{\rm cc}$ are preset positive weights and $\mathcal{T}$ is the time delay parameter set of interest.
 In the objective function of model \eqref{problem_model}, the first term $e(\alpha,\mathbf{X}) $ represents the mismatching square error between the designed beampattern and the desired beampattern $\bar{P}_{\theta}$ and $\alpha$ is a scaling factor that needs to be optimized. The second term $P_c(\mathbf{X})$ relates to the auto-correlation sidelobes and cross-correlation levels at the considered spacial directions. Because $P_{\theta_i,\theta_j,-n}^* = P_{\theta_i,\theta_j,n}$, correlation levels for $n<0$ are not included.
It is difficult to solve \eqref{problem_model} directly since its quartic objective function and constant modulus constraints are non-convex. In the following, by exploiting its inherent structure, we show how to design an efficient solving algorithm to pursue theoretically-guaranteed solutions.

First, let $\mathbf{X}$'s phase be new variable. Since  $x_{i,m}=e^{j\phi_{i,m}}$, we can drop constant modulus constraints and rewrite \eqref{problem_model} as the following minimization problem
\begin{equation}\label{problem_without_constraint}
\begin{split}
&\hspace{0.4cm} \min_{\alpha,\mathbf{\Phi}} \hspace{0.45cm} ~~e\left(\alpha,\mathbf{X}(\mathbf{\Phi})\right) + P_c\left(\mathbf{X}(\mathbf{\Phi})\right),\\
&{\rm subject\ to}\ \ \alpha\in(0, \alpha_{\max}], 0\preceq\mathbf{\Phi}\prec2\pi,
\end{split}
\end{equation}
where the constraint $0\preceq\mathbf{\Phi}\prec2\pi$ means all the elements in $\mathbf{\Phi}$ belong to $[0,2\pi)$.

Second, we define the following quantities
\begin{equation}\label{quantities}
\begin{split}
&\mathbf{a}_{\theta,\theta}= {\rm vec}(\mathbf{a}_{\theta}\mathbf{a}_{\theta}^{H}),\ p = \sum_{\theta\in \Theta} \bar{P}_{\theta},\\
&\mathbf{q} = -\sum_{\theta\in \Theta} \bar{P}_{\theta}\mathbf{a}_{\theta,\theta}, \ \mathbf{A} = \sum_{\theta\in \Theta} \mathbf{a}_{\theta,\theta} \mathbf{a}^H_{\theta,\theta}.
\end{split}
\end{equation}
Then, the first term $e(\alpha,\mathbf{X}(\mathbf{\Phi}))$ in \eqref{problem_without_constraint} can be rewritten as
 \begin{equation}\label{definition_new_e}
  e(\alpha,\mathbf{X}(\mathbf{\Phi})) = \mathbf{v}^H(\alpha,\mathbf{\Phi})\mathbf{Q}\mathbf{v}(\alpha,\mathbf{\Phi}),
\end{equation}
where
\begin{subequations}\label{Qv}
\begin{align}
\mathbf{v}(\alpha,\mathbf{\Phi}) &= \left[
                       \begin{array}{l}
                        ~~~~~~~~~~\alpha\\
                        {\rm vec}\left(\mathbf{X}^H(\mathbf{\Phi})\mathbf{X}(\mathbf{\Phi})\right)
                       \end{array}
                       \right], \label{v}\\
\mathbf{Q}& =  \left[
                       \begin{array}{ll}
                        p & \mathbf{q}^H\\
                        \mathbf{q} & \mathbf{A}
                       \end{array}
                       \right]. \label{Q}
\end{align}
\end{subequations}
Third, to let $P_c(\mathbf{X})$ be in a compact expression, we define $K$-by-$K$ matrices set $\{\mathbf{B}_n(\mathbf{\Phi})|n\!\in\!\mathcal{T}\}$, where $K$ is spacial directions of interest, i.e., set $\hat\Theta$'s size. Specifically, when $n=0$,
\[
  \mathbf{B}_n(\mathbf{\Phi})\!=\!\!\!
  \begin{bmatrix} \!\!\!&0&w_{\rm cc}P_{\theta_1,\theta_2,n} &\dotsb &w_{\rm cc}P_{\theta_1,\theta_K,n} \\
  &\!w_{\rm cc}P_{\theta_2,\theta_1,n}&0 &\dotsb &w_{\rm cc}P_{\theta_2,\theta_K,n} \\
  &\!\vdots\! &\!\vdots\! &\!\ddots\! &\vdots \\
  \!&\!w_{\rm cc}P_{\theta_K,\theta_1,n} &w_{\rm cc}P_{\theta_K,\theta_2,n} &\dotsb &0
  \end{bmatrix},
\]
and when $n\neq0$,
\[
  \mathbf{B}_n(\mathbf{\Phi})\!=\!\!\!
  \begin{bmatrix} \!&\!w_{\rm ac}P_{\theta_1,\theta_1,n} &w_{\rm cc}P_{\theta_1,\theta_2,n} &\dotsb &w_{\rm cc}P_{\theta_1,\theta_K,n} \\
  \!&\!\vdots \!&\!\vdots\! &\!\vdots\! &\!\vdots\! \\
  \!&\!w_{\rm cc}P_{\theta_K,\theta_1,n} &\!w_{\rm cc}P_{\theta_K,\theta_2,n} &\!\dotsb\! &\!w_{\rm ac}P_{\theta_K,\theta_K,n}
  \end{bmatrix}.
\]
 Then, $P_c(\mathbf{X}(\mathbf{\Phi}))$ in \eqref{problem_without_constraint} can be rewritten as
\begin{equation}\label{definition_new_p_c}
  P_c(\mathbf{X}(\mathbf{\Phi}))=\sum_{n\in \mathcal{T}} \|\mathbf{B}_n(\mathbf{\Phi})\|_F^2.
\end{equation}
To facilitate the subsequent derivations, we further define
\begin{subequations}
\begin{align}
&h(\alpha,\mathbf{\Phi}) = \mathbf{v}^H(\alpha,\mathbf{\Phi})\mathbf{Q}\mathbf{v}(\alpha,\mathbf{\mathbf{\Phi}}),\label{f_0} \\
&f_n(\mathbf{\Phi}) = \|\mathbf{B}_n(\mathbf{\Phi})\|_F^2, \label{f_n}
\end{align}
\end{subequations}
and introduce a set of auxiliary variables $\{\mathbf{\Phi}_n| n\in \mathcal{T}\}$.
Then, problem \eqref{problem_without_constraint} can be formulated as the following consensus-like problem \cite{Hong_16}
\begin{equation}\label{Consensus}
\begin{split}
&\min_{\alpha\in\mathds{R},\left\{\mathbf{\Phi},\mathbf{\Phi}_n\right\} \in \mathds{R}^{N \times M}} h(\alpha,\mathbf{\Phi}) + \sum_{n\in \mathcal{T}} f_n(\mathbf{\Phi}_n),\\
&\hspace{0.8cm} {\rm subject\ to}\hspace{0.55cm}\mathbf{\Phi} = \mathbf{\Phi}_n,\ \ n\in \mathcal{T}, \\
&\hspace{2.8cm} \alpha\in(0,\alpha_{\rm max}], \ {0\preceq\mathbf{\Phi}\prec2\pi}.
\end{split}
\end{equation}
In comparison with \eqref{problem_without_constraint}, model \eqref{Consensus} allows subfunction $h(\alpha,\mathbf{\Phi})$ or $f_n(\mathbf{\Phi}_n)$ to handle its local variable independently when $\mathbf{\Phi}$ or $\mathbf{\Phi}_n$ is fixed. In the next section, an efficient algorithm, named by consensus-ADMM, is proposed to solve \eqref{Consensus} approximately. Moreover, we prove that consensus-ADMM converges to a stationary point of model \eqref{problem_without_constraint}\footnotemark.
\footnotetext{Several state-of-the-art algorithms \cite{Kerahroodi_17}--\hspace{-0.01cm}\cite{Aubry_18}, can be customized to handle problem \eqref{problem_model} (not direct). The algorithms keep objective function's value monotonic decreasing or increasing in the iteration procedure. This monotonic property along with some mild conditions is exploited to prove the convergence of the algorithms. Specifically, in \cite{Aubry_18}, the presented algorithm, similar to our proposed consensus-ADMM approach, can also be implemented in parallel. Applying these algorithms to solving \eqref{problem_model} can be  interesting research directions.}
To the best of our knowledge, it is the first time that a parallel algorithm structure is introduced to match the desired beampattern for the MIMO radar system, which means that the proposed consensus-ADMM algorithm is more suitable for the large scale MIMO radar waveforms design problem.
Moreover, convergence analysis and improved variants of the proposed consensus-ADMM algorithm are also considered.

\section{Consensus-ADMM SOLVING ALGORITHM}
\label{sec:pagestyle}

 The augmented Lagrangian function of problem \eqref{Consensus} can be written as
\begin{equation}\label{AL}
\begin{split}
\!\!\!&\mathcal{L}(\alpha,\mathbf{\Phi},\{\mathbf{\Phi}_n,\mathbf{\Lambda}_n,n\in\mathcal{T}\}) \\
\!\!\!=& h(\alpha,\mathbf{\Phi}) \!+\!\!
 \!\sum_{n\in \mathcal{T} } \!\! \left(\!f_n({\mathbf{\Phi}_n})
\! +\! \langle \mathbf{\Lambda}_n,\!\mathbf{\Phi}_n\!\!-\!\mathbf{\Phi}\rangle\!+\! \!\frac{\rho_n}{2}\|\mathbf{\Phi}_n\!\!-\!\mathbf{\Phi}\|_F^2 \!\right),
\end{split}
\end{equation}
where $\mathbf{\Lambda}_n$ and $\rho_n$ are the Lagrangian multiplier and penalty parameters respectively. To facilitate discussions later, we define the following functions
\begin{equation}\label{AL n}
\begin{split}
&\mathcal{L}_n\left(\mathbf{\Phi},\mathbf{\Phi}_n,\mathbf{\Lambda}_n\right) \\
=& f_n({\mathbf{\Phi}_n})
 + \langle\mathbf{\Lambda}_n,\mathbf{\Phi}_n-\mathbf{\Phi}\rangle+ \frac{\rho_n}{2}\|\mathbf{\Phi}_n-\mathbf{\Phi}\|_F^2, \ n\in\mathcal{T}.
\end{split}
 \end{equation}

 Based on \eqref{AL} and \eqref{AL n}, the proposed consensus-ADMM algorithm \cite{Hong_16} can be described as
\begin{subequations}\label{ADMM ori}
\begin{align}
&\{\alpha^{k+1}\!,\mathbf{\Phi}^{k+1}\}\!=\!\!\!\!\underset{\alpha\in(0, \alpha_{\max}],\atop{{0\preceq\mathbf{\Phi}\prec2\pi}}} {\arg \min}\!\! \mathcal{L}\!\left(\alpha,\!\mathbf{\Phi}, \! \{\mathbf{\Phi}_n^k,\mathbf{\Lambda}_n^k,n\!\in\!\mathcal{T}\}\right),\label{step1 ADMM ori}\\
&\mathbf{\Phi}_n^{k+1}=  \underset{\mathbf{\Phi}_n} {\arg \min}\ \ \mathcal{L}_n\left(\mathbf{\Phi}^{k+1}, \mathbf{\Phi}_n,\mathbf{\Lambda}_n^k\right),\  n\in\mathcal{T}, \label{step2 ADMM ori}\\
&\mathbf{\Lambda}_n^{k+1}= \mathbf{\Lambda}_n^{k} + \rho_n(\mathbf{\Phi}_n^{k+1} -\mathbf{\Phi}^{k+1}), \ \ \ \ \ \ \ n\in\mathcal{T},\label{step3 ADMM ori}
\end{align}
\end{subequations}
where $k$ is the iteration number.

{\it Remarks on \eqref{ADMM ori}:} since  $\mathcal{L}\left(\alpha,\mathbf{\Phi}, \{\mathbf{\Phi}_n^k,\mathbf{\Lambda}_n^k, n\in\mathcal{T}\}\right)$ and $\mathcal{L}_n\left(\mathbf{\Phi}^{k+1}, \mathbf{\Phi}_n,\mathbf{\Lambda}_n^k\right)$ are non-convex, it is difficult to implement \eqref{step1 ADMM ori} and \eqref{step2 ADMM ori}. However, we have the following lemma to characterize Lipschitz properties of $\nabla h(\alpha,\mathbf{\Phi})$ and $\nabla f_n(\mathbf{\Phi})$ (see proof in Appendix A).
\begin{lemma}\label{Lipschtiz continuous}
 gradients $\nabla h(\alpha,\mathbf{\Phi})$ and $\nabla f_n(\mathbf{\Phi})$ are Lipschitz continuous, i.e.,
\begin{subequations}\label{Lipschitiz}
\begin{align}
& \|\nabla_{\alpha} h(\alpha, \mathbf{\Phi})\!-\!\nabla_{\alpha} h(\hat{\alpha}, \mathbf{\Phi})\|_F\! \leq\! L_{\alpha} |\alpha\!-\!\hat{\alpha}|,\label{Lip alpha} \\
& \|\nabla_{\mathbf{\Phi}} h(\alpha, \mathbf{\Phi})\!-\!\nabla_{\mathbf{\Phi}} h(\alpha, \hat{\mathbf{\Phi}})\|_F\! \leq\! L \|\mathbf{\Phi}\!-\!\hat{\mathbf{\Phi}}\|_F\!, \label{Lip phi}\\
&\|\nabla f_n({\mathbf{\Phi}})\!-\!\nabla f_n( \hat{\mathbf{\Phi}})\|_F\!\leq \! L_n\|{\mathbf{\Phi}}\!-\!\hat{\mathbf{\Phi}}\|_F, \ n\in \mathcal{T},\label{Lip phi n}
  \end{align}
\end{subequations}
where constants
\begin{subequations}\label{L_ineqs}
\begin{align}
&L_{\alpha}\!\geq\! 2p, \label{L_ineqs_a}\\
&L\geq4(M\!-\!1)(\alpha_{\rm max}\bar{P}_{\rm max}\!+\!M^2N\!+\!2M\!-\!2)|\Theta|, \label{L_ineqs_b} \\
&L_n \geq  2w_{\rm c}^2(2M\!-\!1)(\!M^2N\!+\!2M\!-\!1)K^2. \label{L_ineqs_c}
\end{align}
\end{subequations}
Here, $\bar{P}_{\rm max}=\max\limits_{\theta\in\Theta}\{\bar{P}_{\theta}\}$, $w_{\rm c}= \max\{w_{\rm ac},w_{\rm cc}\}$, and $|\Theta|$ denotes the size of set $\Theta$.
\end{lemma}

Based on the above Lemma, $\mathcal{L}\left(\alpha,\mathbf{\Phi}, \{\mathbf{\Phi}_n^k,\mathbf{\Lambda}_n^k, n\in\mathcal{T}\}\right)$ and $\mathcal{L}_n\left(\mathbf{\Phi}^{k+1}, \mathbf{\Phi}_n,\mathbf{\Lambda}_n^k\right)$ can be upper-bounded by the following strongly convex functions \cite{Bertsekas_99}
\begin{subequations}
  \begin{align}
    \!\!\!\!\!\!\!\!&\mathcal{L}(\alpha,\!\mathbf{\Phi}\!,\! \{\mathbf{\Phi}_n^k,\!\mathbf{\Lambda}_n^k,\!n\in\mathcal{T}\})\!\leq\! \mathcal{U}(\alpha,\!\mathbf{\Phi},\!\{\mathbf{\Phi}_n^k,\!\mathbf{\Lambda}_n^k,\!n\in\mathcal{T}\}),\\
    &\mathcal{L}_n\left(\mathbf{\Phi}^{k+1}, \mathbf{\Phi}_n,\mathbf{\Lambda}_n^k\right)
\leq\mathcal{U}_n(\mathbf{\Phi}^{k+1}, \mathbf{\Phi}_n, \mathbf{\Lambda}_n^k),\ n\in\mathcal{T},
  \end{align}
\end{subequations}
where
 \begin{equation}\label{upperbound U0}
\begin{split}
 \!\!\!&\hspace{0.5cm}\mathcal{U}(\alpha,\mathbf{\Phi}, \{\mathbf{\Phi}_n^k, \mathbf{\Lambda}_n^k, n\!\in\!\mathcal{T}\}) \\
 \!\!\!&\triangleq h(\alpha^k,{\mathbf{\Phi}}^{k})+\langle \nabla_{\mathbf\Phi} h(\alpha^k, {\mathbf{\Phi}}^{k}),\mathbf\Phi-\mathbf\Phi^k\rangle \\
 \!\!\!\hspace{-0.2cm}&+\!\langle \nabla_{\alpha} h(\alpha^k\!, {\mathbf{\Phi}}^{k}),\alpha\!-\!\alpha^k\rangle\!+\!\frac{L}{2}\|{\mathbf{\Phi}}\!-\!{\mathbf{\Phi}}^{k}\|_F^2\!+\!\frac{L_{\alpha}}{2}|\alpha\!-\!\alpha^{k}|^2\\
\!\!\! &+ \sum_{n\in \mathcal{T}}\mathcal{L}_n(\mathbf{\Phi},\mathbf{\Phi}_n^k,\mathbf{\Lambda}_n^k),
\end{split}
\end{equation}
and
\begin{equation}\label{upperbound Un}
\begin{split}
  &\hspace{0.5cm}\mathcal{U}_n(\mathbf{\Phi}^{k+1}, \mathbf{\Phi}_n, \mathbf{\Lambda}_n^k)\\
&\triangleq f_n({\mathbf{\Phi}}^{k+1})+\langle \nabla f_n({\mathbf{\Phi}}^{k+1})+\mathbf{\Lambda}_n^k,{\mathbf{\Phi}_n}-{\mathbf{\Phi}}^{k+1}\rangle \\
&\ \ + \frac{\rho_n+L_n}{2}\|\mathbf{\Phi}_n-\mathbf{\Phi}^{k+1}\|_F^2.
\end{split}
\end{equation}
Then, instead of solving \eqref{step1 ADMM ori} and \eqref{step2 ADMM ori} directly, we propose the following customized consensus-ADMM algorithm \eqref{relax ADMM}.
\begin{subequations}\label{relax ADMM}
\begin{align}
&\{\alpha^{k+1}\!,\mathbf{\Phi}^{k+1}\}\!=\!\!\!\!\underset{\alpha\in(0, \alpha_{\max}],\atop{{0\preceq\mathbf{\Phi}\prec2\pi}}} {\arg \min}\! \mathcal{U}\!\left(\alpha,\!\mathbf{\Phi}, \! \{\mathbf{\Phi}_n^k,\mathbf{\Lambda}_n^k,n\in\mathcal{T}\}\right),\label{relax ADMM a}\\
&\mathbf{\Phi}_n^{k+1}=  \underset{\mathbf{\Phi}_n} {\arg \min} \ \ \mathcal{U}_n\left(\mathbf{\Phi}^{k+1},\mathbf{\Phi}_n,\mathbf{\Lambda}_n^k\right),\label{relax ADMM b}\\
&\mathbf{\Lambda}_n^{k+1} = \mathbf{\Lambda}_n^{k} + \rho_n(\mathbf{\Phi}_n^{k+1} -\mathbf{\Phi}^{k+1}).\label{relax ADMM c}
\end{align}
\end{subequations}

Notice that $\mathcal{U}(\alpha,\mathbf{\Phi}, \{\mathbf{\Phi}_n^k, \mathbf{\Lambda}_n^k, n\!\in\!\mathcal{T}\})$ is a strongly convex quadratic function with respect to $\alpha$ and $\mathbf{\Phi}$ respectively and $\mathcal{U}_n(\mathbf{\Phi}^{k+1}, \mathbf{\Phi}_n, \mathbf{\Lambda}_n^k)$ is also some strongly convex quadratic function with respect to $\mathbf{\Phi}_n$. Therefore, the minimizers of $\mathcal{U}(\alpha,\mathbf{\Phi}, \{\mathbf{\Phi}_n^k, \mathbf{\Lambda}_n^k, n\!\in\!\mathcal{T}\})$ and $\mathcal{U}_n(\mathbf{\Phi}^{k+1}, \mathbf{\Phi}_n, \mathbf{\Lambda}_n^k)$ can be determined through the following procedures:

Set the gradients of the functions $\mathcal{U}(\alpha,\mathbf{\Phi}, \{\mathbf{\Phi}_n^k, \mathbf{\Lambda}_n^k, n\!\in\!\mathcal{T}\})$ and $\mathcal{U}_n(\mathbf{\Phi}^{k+1}, \mathbf{\Phi}_n, \mathbf{\Lambda}_n^k)$ to be zeros
\begin{subequations}
\begin{align}
& \nabla_{\alpha} \mathcal{U}(\alpha,\mathbf{\Phi}, \{\mathbf{\Phi}_n^k, \mathbf{\Lambda}_n^k, n\!\in\!\mathcal{T}\}) \!= \!0,\\
&  \nabla_{\mathbf{\Phi}} \mathcal{U}(\alpha,\mathbf{\Phi}, \{\mathbf{\Phi}_n^k, \mathbf{\Lambda}_n^k, n\!\in\!\mathcal{T}\}) =\!0,\\
 &\nabla_{\mathbf{\Phi}_n} \mathcal{U}_n(\mathbf{\Phi}^{k+1}, \mathbf{\Phi}_n, \mathbf{\Lambda}_n^k)=\!0,
\end{align}
\end{subequations}
which lead to the following linear equations
\begin{subequations}\label{solutions}
\begin{align}
&\nabla_{\alpha} h(\alpha^k,{\mathbf{\Phi}}^{k})\!+ \! L_{\alpha}(\alpha\!-\!\alpha^{k})\!=\!0,\\
& \nabla_{\mathbf\Phi} h(\!\alpha^k, {\mathbf{\Phi}}^{k}\!)\!+\!L(\!{\mathbf{\Phi}}\!-\!{\mathbf{\Phi}}^{k}\!)\!-\!\!\!\sum_{n\in \mathcal{T}}\!\!\left(\!\mathbf{\Lambda}_n^k\!+\!\rho_n({\mathbf{\Phi}}_n^k\!-\!{\mathbf{\Phi}}\!)\right)\!\!=\!0,\\
&\nabla f_n({\mathbf{\Phi}}^{k+1})\!+\!\mathbf{\Lambda}_n^k  \!+\! (\rho_n+L_n)(\mathbf{\Phi}_n-\mathbf{\Phi}^{k+1})\!=\!0.
\end{align}
\end{subequations}
Then, by solving the above linear equations and projecting the solutions onto the corresponding feasible regions, we can obtain
\begin{subequations}\label{solutions}
\begin{align}
&\alpha^{k+1} = \underset{(0,\alpha_{
\rm max]}}\Pi \left(\alpha^k-\frac{\nabla_{\alpha}h(\alpha^k,\mathbf\Phi^k)}{L_\alpha}\right), \label{solution alpha}\\
&\mathbf\Phi^{k+1}\!\!=\!\!\!\underset{[0,2\pi)}\Pi\!\! \left(\! \frac{L\mathbf\Phi^k\!\!-\!\!\nabla_{\mathbf\Phi}h(\alpha^k\!,\!\mathbf\Phi^k) \!\!+\!\!\!\!\underset{{n\in\mathcal{T}}}\sum\!(\mathbf\Lambda_n^k\!+\!\rho_n\mathbf\Phi_n^k)}{L+\underset{{n\in\mathcal{T}}}\sum\rho_n}\!\right)\!, \label{solution phi}\\
&\mathbf\Phi_n^{k+1} = \mathbf\Phi^{k+1} -\frac{\nabla f_n(\mathbf\Phi^{k+1})+\mathbf{\Lambda}_n^k}{\rho_n+L_n}. \label{solution phi n}
\end{align}
\end{subequations}

Combining \eqref{relax ADMM c} and \eqref{solutions}, we summarize the proposed consensus-ADMM algorithm in Table \ref{proposed ADMM algorithm}.
\begin{table}[t]
\renewcommand \arraystretch{1.2}
\caption{The proposed consensus-ADMM algorithm }
\label{proposed ADMM algorithm}
\centering
\begin{tabular}{l}
 \hline\hline
 {\bf Initialize:} compute Lipschitz constants $\{L_n,\!n\!\in\!\mathcal{T}\}$ and \\\hspace{0.2cm} $L$. Set iteration index $k=1$, initialize $\mathbf{\Phi}^1$ and $\{\mathbf\Lambda_n^1\}$, \\ \hspace{0.2cm} and let $\{\mathbf{\Phi}^1 = \mathbf{\Phi}_n^1, n\in \mathcal{T}\}$. \\
  {\bf Repeat} \\
  \hspace{0.2cm} Step 1: compute $\alpha^{k+1}$ and $\mathbf\Phi^{k+1}$ via \eqref{solution alpha} and \eqref{solution phi}. \\
  \hspace{0.2cm} Step 2: update $\{\mathbf\Phi_n^{k+1}, n\in\mathcal{T}\}$ and $\{\mathbf{\Lambda}_n^{k+1}, n\in\mathcal{T}\}$\\
  \hspace{1.4cm}  via \eqref{solution phi n} and \eqref{relax ADMM c} respectively in parallel. \\
 {\bf Until} some preset termination criterion is satisfied.\\
 \hline\hline
\end{tabular}
\end{table}

\section{Analysis}

\subsection{Convergence Issue}

We have the following theorem to show convergence properties of the proposed consensus-ADMM algorithm in Table I.
\begin{theorem}\label{convergence}
Let $(\alpha^*,\mathbf{\Phi}^{*},\{\mathbf{\Phi}_n^*, \mathbf{\Lambda}_n^{*},n\in\mathcal{T}\})$ denote some limit point of the sequence $\left(\alpha^k,\mathbf{\Phi}^k, \{\mathbf{\Phi}_n^k, \mathbf{\Lambda}_n^k, n\in\mathcal{T}\}\right)$. $\forall n\in \mathcal{T}$, if the penalty parameters $\rho_n$ and Lipschitz constants $L_n$ satisfy $\rho_n\geq 9L_n$, the proposed consensus-ADMM algorithm is convergent, i.e.,
\begin{equation}\label{convergence variables}
\begin{split}
&\lim\limits_{k\rightarrow+\infty}\alpha^{k}=\alpha^{*}, \ \ \lim\limits_{k\rightarrow+\infty}\mathbf{\Phi}^{k}=\mathbf{\Phi}^{*}, \lim\limits_{k\rightarrow+\infty}\mathbf{\Phi}_n^{k}=\mathbf{\Phi}_n^{*},\\
& \lim\limits_{k\rightarrow+\infty}\mathbf{\Lambda}_n^{k}=\mathbf{\Lambda}_n^{*},\ \mathbf{\Phi}^*=\mathbf{\Phi}_n^*,\ \  \forall~n\in \mathcal{T}.
\end{split}
\end{equation}
Moreover, $(\alpha^*,\mathbf{\Phi}^{*})$ is a stationary point of problem \eqref{problem_without_constraint}, i.e., it satisfies the following inequalities
\begin{equation}\label{stationary point}
\begin{split}
\hspace{-0.2cm}&\langle\nabla_{\alpha} e\left(\alpha^*\!,\mathbf{X}(\mathbf{\Phi}^*)\!\right),{\alpha}\!-\!{\alpha^*}\rangle\geq0,\ \alpha\!\in\!(0,\alpha_{\max}],\\
\hspace{-0.2cm}&\langle \nabla_{\!\mathbf{\Phi}} e\!\left(\!\alpha^*\!\!,\mathbf{X}(\!\mathbf{\Phi}^*)\!\right)\!\!+\!\!\nabla_{\!\mathbf{\Phi}}  P_c\!\left(\!\mathbf{X}(\mathbf{\Phi}^*)\!\right)\!,\mathbf{\Phi}\!-\!\mathbf{\Phi}^*\rangle\!\!\geq\!\!0,{0\!\!\preceq\!\mathbf{\Phi}\!\!\prec\!2\pi}.
\end{split}
\end{equation}
\end{theorem}

{\it Remarks:}
   Theorem \ref{convergence} indicates that the proposed consensus-ADMM algorithm is theoretically-guaranteed to be convergent to a  stationary point of model \eqref{problem_without_constraint} under the  conditions $\rho_n\geq 9L_n, n\in\mathcal{T}$. Here, we should note that these conditions are easily satisfied since we can choose $L_n$'s value according to \eqref{L_ineqs_c} in Lemma 1 and penalty parameters $\rho_n$ can be set accordingly to satisfy $\rho_n\geq 9L_n$. The key idea of proving Theorem \ref{convergence} is to find out that potential function $\mathcal{L}(\alpha,\mathbf{\Phi},\{\mathbf{\Phi}_n,\mathbf{\Lambda}_n,n\in\mathcal{T}\})$ {\it decreases sufficiently} in every ADMM iteration and is lower-bounded. To reach this goal, we first prove several related lemmas in Appendix B. Then, we give the detailed proof of Theorem 1 in Appendix C.

\subsection{Implementation Analysis}

In the following, we show how to compute $\nabla_{\alpha} h(\alpha,{\mathbf{\Phi}})$,  $\nabla_{\mathbf\Phi} h(\alpha,{\mathbf{\Phi}})$ and $\nabla f_n({\mathbf{\Phi}})$ efficiently by exploiting their inside structures.

 1) $\nabla_{\alpha} h(\alpha,{\mathbf{\Phi}})$ and $\nabla_{\mathbf\Phi} h(\alpha,{\mathbf{\Phi}})$. They can be expressed as follows respectively
\begin{subequations}\label{gradient h}
 \begin{align}
&\hspace{-0.2cm}\nabla_{\alpha} h(\alpha,\mathbf{\Phi}) = 2{\rm Re}\bigg(\frac{\partial\mathbf{v}^H(\alpha,\mathbf{\Phi})}{\partial \alpha}\mathbf{Q}\mathbf{v}(\alpha,\mathbf{\Phi})\bigg), \label{gradient alpha} \\
&\hspace{-0.2cm}\nabla_{\mathbf{\Phi}} h(\alpha,\mathbf{\Phi})\! \!=\! {\rm mat}\!\!\left(\!\!2{\rm Re}\bigg(\!\frac{\partial\mathbf{v}^H(\alpha,\mathbf{\Phi})}{\partial {\rm vec}(\mathbf{\Phi})}\!\mathbf{Q}\mathbf{v}(\alpha,\!\mathbf{\Phi})\!\!\bigg),\!N,\!M\!\!\right),\label{gradient phi_0}
\end{align}
\end{subequations}
where $\frac{\partial\mathbf{v}(\alpha,\mathbf{\Phi})}{\partial \alpha}=\left[1;\mathbf{0}\right]$ and
\begin{equation}\label{patial phi}
 \frac{\partial\mathbf{v}^H(\alpha,\mathbf{\Phi})}{\partial {\rm vec}(\mathbf{\Phi})}\!\!=\!\!\displaystyle\left[\! \frac{\partial\mathbf{v}(\alpha,\mathbf{\Phi})}{\partial \phi_{1,1}},\! \frac{\partial\mathbf{v}(\alpha,\mathbf{\Phi})}{\partial \phi_{2,1}},\! \dotsb,\! \frac{\partial\mathbf{v}(\alpha,\mathbf{\Phi})}{\partial \phi_{N,M}}\!\right]^H.
\end{equation}
In \eqref{patial phi}, $\frac{\partial\mathbf{v}(\alpha,\mathbf{\Phi})}{\partial \phi_{i,m}}$ can be calculated through
 \begin{equation}\label{gradient phi_m}
\begin{split}
\frac{\partial\mathbf{v}(\alpha,\mathbf{\Phi})}{\partial \phi_{i,m}} = \left[
\begin{array}{l}
~~~~~~~~~~~~~~0\\ {\rm vec}\left(\displaystyle\frac{\partial\mathbf{X}^H(\mathbf{\Phi})\mathbf{X}(\mathbf{\Phi}) }{\partial \phi_{i,m}}\right)
\end{array}
 \right],
\end{split}
\end{equation}
where $i=1,\dotsb,N$, $m=1,\dotsb,M$, and $\frac{\partial (\mathbf{X}^H(\mathbf{\Phi})\mathbf{X}(\mathbf{\Phi}))}{\partial\phi_{i,m}}$ can be computed through \eqref{partial derivative}.
\begin{figure*}

\begin{equation}\label{partial derivative}
\begin{split}
 \frac{\partial \mathbf{X}^{\!H}(\mathbf{\Phi})\mathbf{X}(\mathbf{\Phi})}{\partial\phi_{i,m}}\!\!=\!\!
\left[\!\!
  \begin{array}{ccccccc}
    ~ &  ~& ~ & je^{j(\!\phi_{i,m}\!-\phi_{i,1}\!)} & ~ & ~ & ~ \\
    ~ &  \! \! \mathbf{\scalebox{3.5}0}\!\! & ~& \vdots & ~  & \!\! \mathbf{\scalebox{3.5}0}\!\!  & ~ \\
   ~  &~ &~ & je^{j(\!\phi_{i,m}\!-\phi_{i,m\!-\!1}\!)}&~ &~  \\
    \!-\!je^{j(\!\phi_{i,1}\!-\phi_{i,m}\!)} &\!\!\cdots\!\! &\!-\!je^{j(\phi_{i,m\!-\!1}\!-\phi_{i,m}\!)} & 0 & \!-\!je^{j(\!\phi_{i,m\!+\!1}\!-\phi_{i,m}\!)}&  \!\!\cdots\!\! & \!-\!je^{j(\!\phi_{i,M}\!-\phi_{i,m}\!)}  \\
    ~ & ~ &~ &je^{j(\phi_{i,m}\!-\phi_{i,m\!+\!1}\!)}  &~ & ~& ~  \\
    ~ &  \!\!\mathbf{\scalebox{3.5}0} \!\! & ~ & \vdots & ~ & \!\! \mathbf{\scalebox{3.5}0} \!\! & ~  \\
    ~ & ~ &~ & je^{j(\!\phi_{i,m}\!-\phi_{i,M}\!)} &~ & ~& ~  \\
  \end{array}
\!\!\right].
\end{split}
\end{equation}
\hrulefill
\vspace*{4pt}
\end{figure*}

Since $\mathbf{Q}\in\mathds{C}^{(M^2+1)\times (M^2+1)}$,  $\mathbf{v}\in\mathds{C}^{M^2+1}$, and $\frac{\partial\mathbf{v}(\alpha,\mathbf{\Phi})}{\partial \alpha}=\left[1;\mathbf{0}\right]$, we can obtain $\frac{\partial\mathbf{v}^H(\alpha,\mathbf{\Phi})}{\partial \alpha)}$ through no more than $(M^2+1)^2$ complex multiplications. Moreover, since $\frac{\partial\mathbf{X}^H(\mathbf{\Phi})\mathbf{X}(\mathbf{\Phi})}{\partial \phi_{i,m}}$ involves $2(M-1)$ nonzero elements, then there are $2(M-1)MN$ nonzero elements in $\frac{\partial\mathbf{v}^H(\alpha,\mathbf{\Phi})}{\partial {\rm vec}(\mathbf{\Phi})}$. It means that it takes $(M^2+1)^2+(2M-1)MN$ complex multiplications to compute $\nabla_{\mathbf{\Phi}} h(\alpha,\mathbf{\Phi})$. Then, we can see that the total computation cost on $\nabla_{\alpha} h(\alpha,{\mathbf{\Phi}})$ and $\nabla_{\mathbf\Phi} h(\alpha,{\mathbf{\Phi}})$ is roughly $\mathcal{O}(M^4+2M^2N)$.

2) $\nabla f_n({\mathbf{\Phi}})$. The elements in $\nabla f_n({\mathbf{\Phi}})$ can be obtained through
 \begin{equation}\label{gradient phi_n}
\frac{\partial f_n({\mathbf{\Phi}})}{\partial \phi_{i,m}}\!\! =\!\! \left\{\!\!\!\!
\begin{array}{l}
\displaystyle\sum\limits_{{\theta_i}\neq{\theta_j}\atop\theta_i,\theta_j \in \hat{\Theta}}\!\!\!2w_{\rm cc}{\rm Re}\!\left(\!\!
P_{\theta_i,\theta_j,n}^*\frac{\partial P_{\theta_i,\theta_j,n} }{\partial \phi_{i,m}}\!\!\right)\!, n=0,\\
\displaystyle\sum \limits_{\theta_i,\theta_j\in \hat{\Theta}}\!\!2\bigg(\!\!w_{\rm ac}{\rm Re}\!\left(\!\!
P_{\theta_i,\theta_i,n}^*\!\!\frac{\partial P_{\theta_i,\theta_i,n} }{\partial \phi_{i,m}}\!\!\right)\!\\
\hspace{0.7cm} + w_{\rm cc}{\rm Re}\!\left(\!\!
P_{\theta_i,\theta_j,n}^*\!\!\frac{\partial P_{\theta_i,\theta_j,n} }{\partial \phi_{i,m}}\!\!\right)\!\!\bigg),n\in\!\! \mathcal{T}\backslash0.
\end{array}
\right.
\end{equation}

To compute $P_{\theta_i,\theta_j,n}$ and $\frac{\partial P_{\theta_i,\theta_j,n} }{\partial \phi_{i,m}}$ for every $\frac{\partial f_n({\mathbf{\Phi}})}{\partial \phi_{i,m}}$ efficiently, we define $\mathbf{s}_{\theta_i}=\mathbf{X}\mathbf{a}_{\theta_i}$ and $\mathbf{\bar{s}}_{\theta_i}$ denoting $\mathbf{s}_{\theta_i}$'s reversing vector. Then, $\forall n\in\mathcal{T}$, since $P_{\theta_i,\theta_j,n} = \mathbf{a}_{\theta_i}^{H}\mathbf{X}^H\mathbf{S}_n\mathbf{X}\mathbf{a}_{\theta_j}$, it can be obtained through convolution operation
$
  \mathbf{s}_{\theta_i}^*\otimes\mathbf{\bar{s}}_{\theta_j}.
$
It means that the cost to obtain all the $P_{\theta_i,\theta_j,n}$ is roughly $|\mathcal{T}|K^2 MN $ complex multiplications.
Moreover, corresponding gradients $\frac{\partial P_{\theta_i,\theta_j,n} }{\partial \phi_{i,m}}$ can be obtained through
$
\mathbf{s}_{\theta_i}^*\otimes\frac{\partial \bar{\mathbf{s}}_{\theta_j}}{\partial \phi_{i,m}} + \bar{\mathbf{s}}_{\theta_j}\otimes\frac{\partial \mathbf{s}_{\theta_i}^*}{\partial \phi_{i,m}}.
$
Since there is only one nonzero element in either $\frac{\partial \bar{\mathbf{s}}_{\theta_j}}{\partial \phi_{i,m}}$ or $\frac{\partial \mathbf{s}_{\theta_i}^*}{\partial \phi_{i,m}}$, it takes only two complex multiplications to compute $\frac{\partial P_{\theta_i,\theta_j,n} }{\partial \phi_{i,m}}$. Then, we can obtain all of them through $2|\mathcal{T}|K^2MN$ complex multiplications.
Taking complex multiplication operations to obtain $P_{\theta_i,\theta_j,n}^*\frac{\partial P_{\theta_i,\theta_j,n} }{\partial \phi_{i,m}}$ into account, we can see that the total cost of computing $\frac{\partial f_n({\mathbf{\Phi}})}{\partial \phi_{i,m}}$ $\forall n\in\mathcal{T}$ is roughly $\mathcal{O}(3|\mathcal{T}|K^2MN)$.

Observing the proposed consensus-ADMM algorithm in Table I and the corresponding \eqref{solutions} and \eqref{relax ADMM c}, we can see that the main computational cost lies in computing $\nabla_{\alpha} h(\alpha,{\mathbf{\Phi}})$, $\nabla_{\mathbf\Phi} h(\alpha,{\mathbf{\Phi}})$ and $\nabla f_n({\mathbf{\Phi}})$, which are much larger than other terms. Therefore, we conclude that the total cost in each ADMM iteration is roughly $\mathcal{O}(M^4+2M^2N+3|\mathcal{T}|K^2MN)$.

\section{Improvements}
\subsection{Reduce Complexity}
In the proposed consensus-ADMM algorithm, $\mathbf{\Phi}_{n}$, $\forall n\in \mathcal{T}$, are updated independently (or in parallel) as are the Lagrangian multipliers $\mathbf{\Lambda}_n$.
This fact admits us to update only a part of the variables  $\{\mathbf{\Phi}_{n}, \mathbf{\Lambda}_n, n\in\mathcal{T}\}$ in each ADMM iteration to reduce computational complexity.

Specifically, consider a randomized updating strategy called stochastic block coordinate descent (SBCD) \cite{Tseng_01}. In the $k$-th iteration, let $\mathcal{N}^{k}$ denote some $\mathcal{T}$'s subset. We choose elements from $\mathcal{T}$ to construct $\mathcal{N}^{k}$ with the probability
\begin{equation}\label{update probability}
\begin{split}
{\rm Pr}(n\in\mathcal{N}^{k}) = p_n \geq p_{\rm min}>0. \footnotemark
\end{split}
\end{equation}
If some $n\in \mathcal{N}^k$, the corresponding variables $\mathbf{\Phi}_n^k$ and $\mathbf{\Lambda}_n^k$ are updated using \eqref{solution phi n} and \eqref{relax ADMM c} respectively.
Otherwise, we just set
 $
 \mathbf{\Phi}_{n}^{k+1} = \mathbf{\Phi}_{n}^{k},
\mathbf{\Lambda}_n^{k+1} = \mathbf{\Lambda}_n^{k}.
$

\footnotetext{Usually, criteria of selecting $\mathcal{N}^k$ is to guarantee every element in $\mathcal{T}$ is implemented equally in probability.}

In this way, it is obvious that computational complexity in each ADMM iteration can be reduced significantly.
The approach provides an option for some practical systems when their computation resources are very limited.
Moreover, this kind of implementation strategy can still guarantee the algorithm converges with high probability to a stationary point of problem \eqref{problem_without_constraint} under some wild conditions (we provide a sketch of the proof in Appendix D).

\subsection{Speed Up Convergence}
Besides the computational complexity in each iteration, convergence speed is another concern from a practical viewpoint. In this paper,
inspired by Nesterov's accelerated gradient descent method (AGD) \cite{Nesterov_83}, which is originally applied to a convex problem, we develop its variant in the following
\begin{subequations}\label{Nesterov-acc}
\begin{align}
&\mathbf{\hat{\Phi}}_{n}^{k+1} = \underset{\mathbf{\Phi}_n} {\arg \min} \ \ \mathcal{U}_n\left(\mathbf{\Phi}^{k+1},\mathbf{\Phi}_n,\mathbf{\Lambda}_n^k\right), \label{Nesterov-acc-1} \\
&\mathbf{\Phi}_{n}^{k+1}=\mathbf{\hat{\Phi}}_{n}^{k+1}+\gamma^k\left(\mathbf{\hat{\Phi}}_{n}^{k+1}-\mathbf{\hat{\Phi}}_{n}^{k}\right), \label{Nesterov-acc-2}
\end{align}
\end{subequations}
where $\gamma^k =\frac{k-1}{k+t-1}$ and $t \geq 3$ is some preset constant.
The algorithm starts from $\mathbf{\hat{\Phi}}_{n}^{1} = \mathbf{\Phi}_{n}^{1}$.
Here, we should note that AGD method's convergence can be proved under strong assumptions, such as solving convex or strongly convex optimization problem \cite {Nesterov_04}--\hspace{-0.001cm}\cite{Goldstein_14}. Here, the direct combination of consensus-ADMM and AGD can be cast as a heuristic method. It is difficult to prove that it can improve convergence rate theoretically for the considered non-convex probing waveform design problem \cite{Goldstein_14}--\hspace{-0.001cm}\cite{botao_global_convergence}.
 However, simulation results,  presented in the following section, show that its practical performance is superior to the original one in Table I.

\section{Simulation results}\label{Simulation results}
\begin{figure}[htbp]
\centering
  \centerline{\includegraphics[scale=0.65]{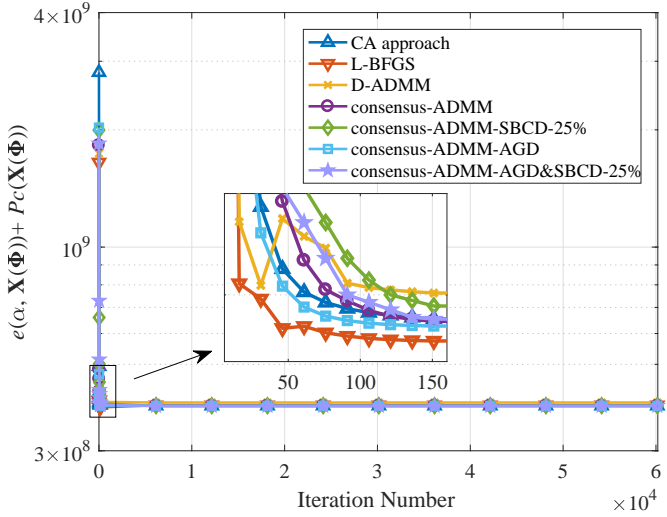}}
\centerline{(a) $e(\alpha,\mathbf{X}(\mathbf{\Phi}))+P_c(\mathbf{X}(\mathbf{\Phi}))$ vs. iteration number. }
  \centering
  \centerline{\includegraphics[scale=0.65]{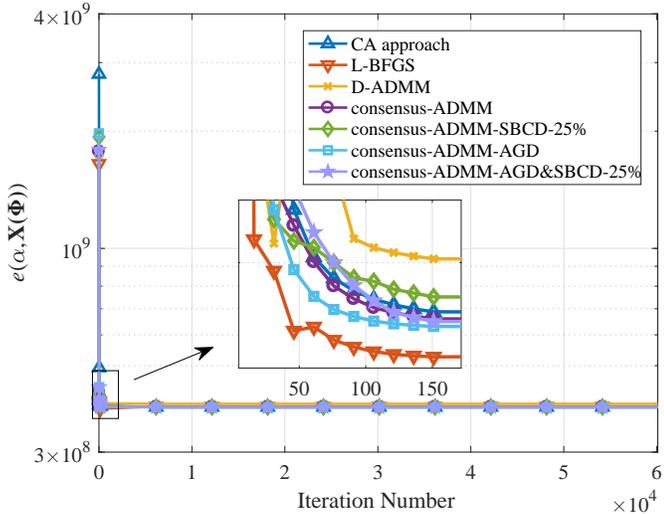}}
  \centerline{(b) $e(\alpha,\mathbf{X}(\mathbf{\Phi}))$ vs. iteration number.}
    \centering
  \centerline{\includegraphics[scale=0.65]{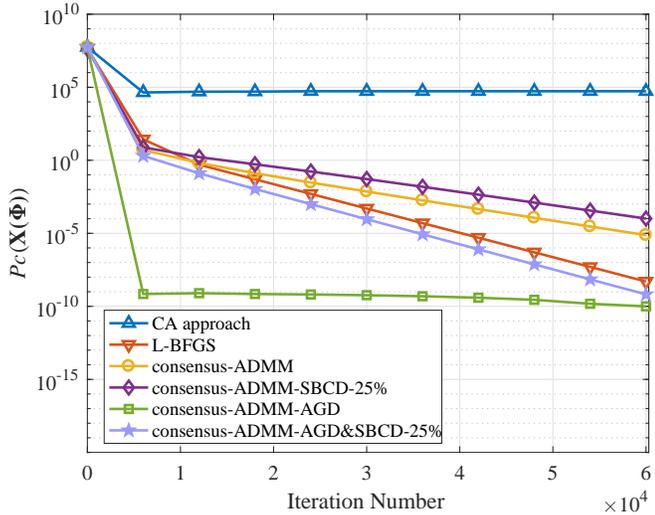}}
  \centerline{(c) $ P_c(\mathbf{X}(\mathbf{\Phi}))$ vs. iteration number.} \medskip
\caption{ Comparisons of convergence performance with $M=8,N=128,\mathcal{T}=[0,16]$, SBCD-$25\%$ means that one fourth of the elements in set $\mathcal{T}$ are updated.}
\label{convergence fig}
\end{figure}
In this section, numerical results are presented to illustrate the performance of the proposed MIMO radar beampattern design algorithm.
We consider the number of antennas as $M = 8,16,128$ with the length of each sequence $N = 64,128,1024$ respectively.
The set of the spacial angles covers $(-90^\circ,90^\circ)$ with spacing $0.1^\circ$.
The residuals of the proposed consensus-ADMM algorithm in the $k$-th iteration are defined as
$
\displaystyle\sum_{n\in \mathcal{T} }\left\lVert  \mathbf{\Phi}_n^{k+1} \!-\!\mathbf{\Phi}^{k+1}\right\rVert_F$ and
$\displaystyle\sum_{n\in \mathcal{T} }\left\lVert \mathbf{\Phi}_n^{k+1} \!-\!\mathbf{\Phi}_n^{k}\right\rVert_F.
$
The termination criteria is set as both of the residuals are less than $10^{-4}$ or the maximum iteration number 60000 is reached.
The weights $(w_{\rm ac},w_{\rm cc})$ are  $(10,10)$.
The desired beampattern is
\begin{equation}\label{desied beampattern}
\begin{split}
\bar{P}(\theta)=\left\{
           \begin{array}{ll}
             1, & \theta\in[\theta_i-10^\circ,\theta_i+10^\circ],~i=1,2, \\
             0, & {\rm otherwise},
           \end{array}
         \right.
\end{split}
\end{equation}
where $\theta_1=-40^\circ$ and $\theta_2=30^\circ$.
The parameter $t$ in the AGD method is $3$.
The penalty parameter $\rho_n$ can affect ADMM algorithm's convergence rate. For example, much larger $\rho_n$ will let the optimization problem become singular and slow it down. Here, we recommend its value as $\rho_n=\|{\rm vec}(\nabla f_n(\mathbf\Phi^{1}))\|_{\infty}$\footnotemark. For D-ADMM, the penalty parameters are set as 20.
The random phase sequence is chosen to initialize all approaches.
All experiments are performed in MATLAB 2016b/Windows 7 environment on a computer with 2.1GHz Intel 4110$\times$2 CPU and 64GB RAM.
\footnotetext{Here, $\mathbf{\Phi}^{1}$ is initial value of the phase variable.  Using this setting, simulation results are pretty good and convergence can always be observed.}

Figure \ref{convergence fig} plots the performance curves of objective function versus the iteration number for our proposed consensus-ADMM algorithms and other three state-of-the-art approaches: CA approach \cite{Li_08}, L-BFGS method \cite{Wang_12}, and D-ADMM \cite{Cheng_17}. Besides the objective function in \eqref{problem_model}, its two parts: $e(\alpha, \mathbf{X}(\mathbf{\Phi}))$ and $P_c(\mathbf{X}(\mathbf{\Phi}))$ are also presented in the figures since both of them have obvious physical significance. It should be noted that the D-ADMM method only focuses on beampattern matching problem. Its spacial correlation characteristics are not included in the corresponding figure.
 From Figure \ref{convergence fig}, we can see that the resulting $e(\alpha, \mathbf{X}(\mathbf{\Phi}))$ is much larger than $P_c(\mathbf{X}(\mathbf{\Phi}))$.
 Different approaches for $e(\alpha, \mathbf{X}(\mathbf{\Phi}))$ have similar convergence performance. However, convergence results for $P_c(\mathbf{X}(\mathbf{\Phi}))$ are quite different.
 Specifically, the AGD strategy \eqref{Nesterov-acc} can speed up convergence very well. Other algorithms tend to achieve the similar value to AGD method with relatively large number of iterations.
In comparison, consensus-ADMM-SBCD-$25\%$'s convergence is a little bit slow.
However, we should note that it has lower computational complexity. In practice, parameter $25\%$ can be adjusted to make a tradeoff between convergence rate and computational complexity.

\begin{figure*}[htbp]
\centering
\subfigure[ Comparison of the  synthesized beampattern with $M=16, N=64$.]{
\begin{minipage}[t]{0.5\linewidth}
\centering
\centerline{\includegraphics[scale=0.34]{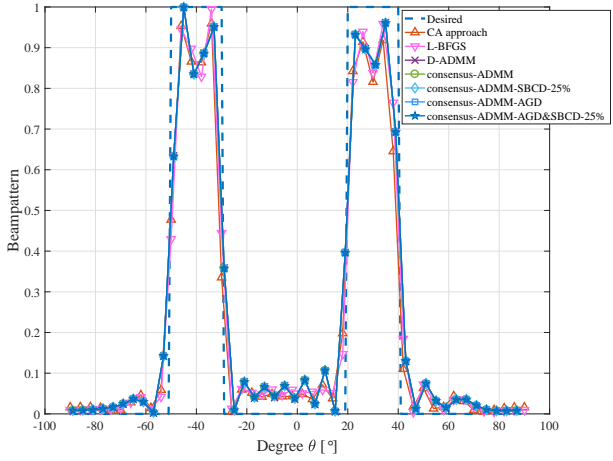}}
\end{minipage}%
}%
\subfigure[ Comparison of the  synthesized beampattern with $M=8, N=128$.]{
\begin{minipage}[t]{0.5\linewidth}
\centering
\centerline{\includegraphics[scale=0.34]{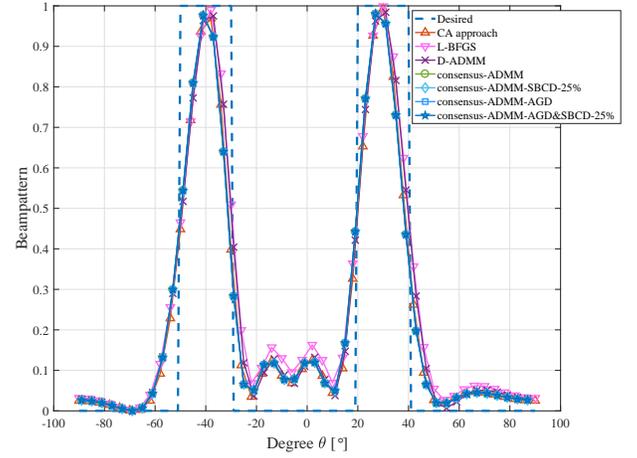}}
\end{minipage}
}
\caption{Comparison of the synthesized beampattern.}
\label{beampattern}
\end{figure*}

Figure \ref{beampattern} shows synthesized spacial beampatterns by our proposed consensus-ADMM approaches and three other approaches.
From the figure, it can be observed that all of the approaches can match the desired spacial beampattern very well at different antenna numbers and waveform lengths.

Figures \ref{correlation_a} -- \ref{correlation_c} show the normalized spacial correlation level $C_{\theta_i,\theta_j,n}$ with different simulation parameters. Here, the normalized spacial correlation function $C_{\theta_i,\theta_j,n}$ for a certain interval in dB is defined as
\begin{equation*}\label{eq_autocorr}
\begin{split}
C_{\theta_i,\theta_j,n}=~&10\log_{10}\frac{|P_{\theta_i,\theta_j,n} |}{\max\{|P_{\theta_i,\theta_i,0}|,|P_{\theta_j,\theta_j,0}|\}},
\end{split}
\end{equation*}
where $n\in \mathcal{T}$ and $\theta_i,\theta_j\in\hat{\Theta}$. It is obvious that for $i = j$ and $i \neq j$, $C_{\theta_i,\theta_i,n}$ indicates spacial auto-/cross-correlation characteristics of the designed MIMO radar probing waveforms respectively. From the figures, we can see that the normalized spacial auto-correlation functions are symmetric and the normalized cross-correlation functions of $(-40^\circ,30^\circ)$ are symmetric to that of $(30^\circ,-40^\circ)$.
The figures also indicate that either increasing $N$ and $M$ or decreasing $\mathcal{T}$ can lower the correlation levels. These facts are reasonable since larger $N$ and $M$ or smaller $\mathcal{T}$ indicate more degrees of freedom in designing probing waveforms. Moreover, we can also see that the consensus-ADMM-AGD approach enjoys the best auto-/cross-correlation characteristics. This fact is in accordance with the simulation result in Figure \ref{convergence fig}.

\begin{figure}[htb]
  \centering
  \centerline{\includegraphics[scale=0.35]{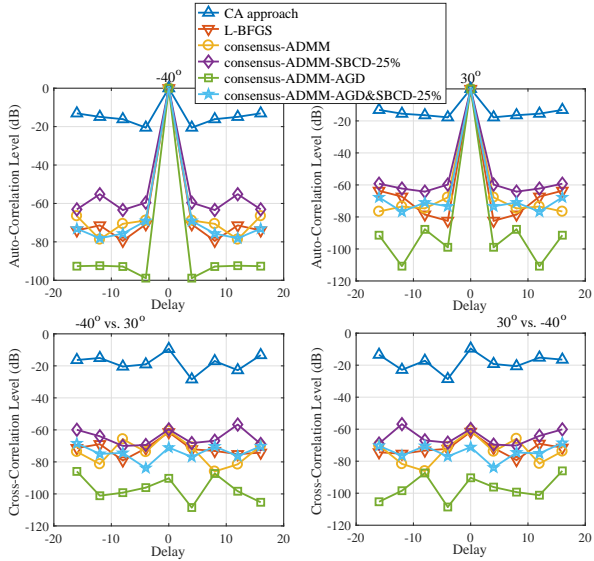}}
\caption{Comparison of correlation characteristics for interval $[0,16]$ with $M=8,N=64$.}
\label{correlation_a}
\end{figure}
\begin{figure}[htb]
  \centering
  \centerline{\includegraphics[scale=0.35]{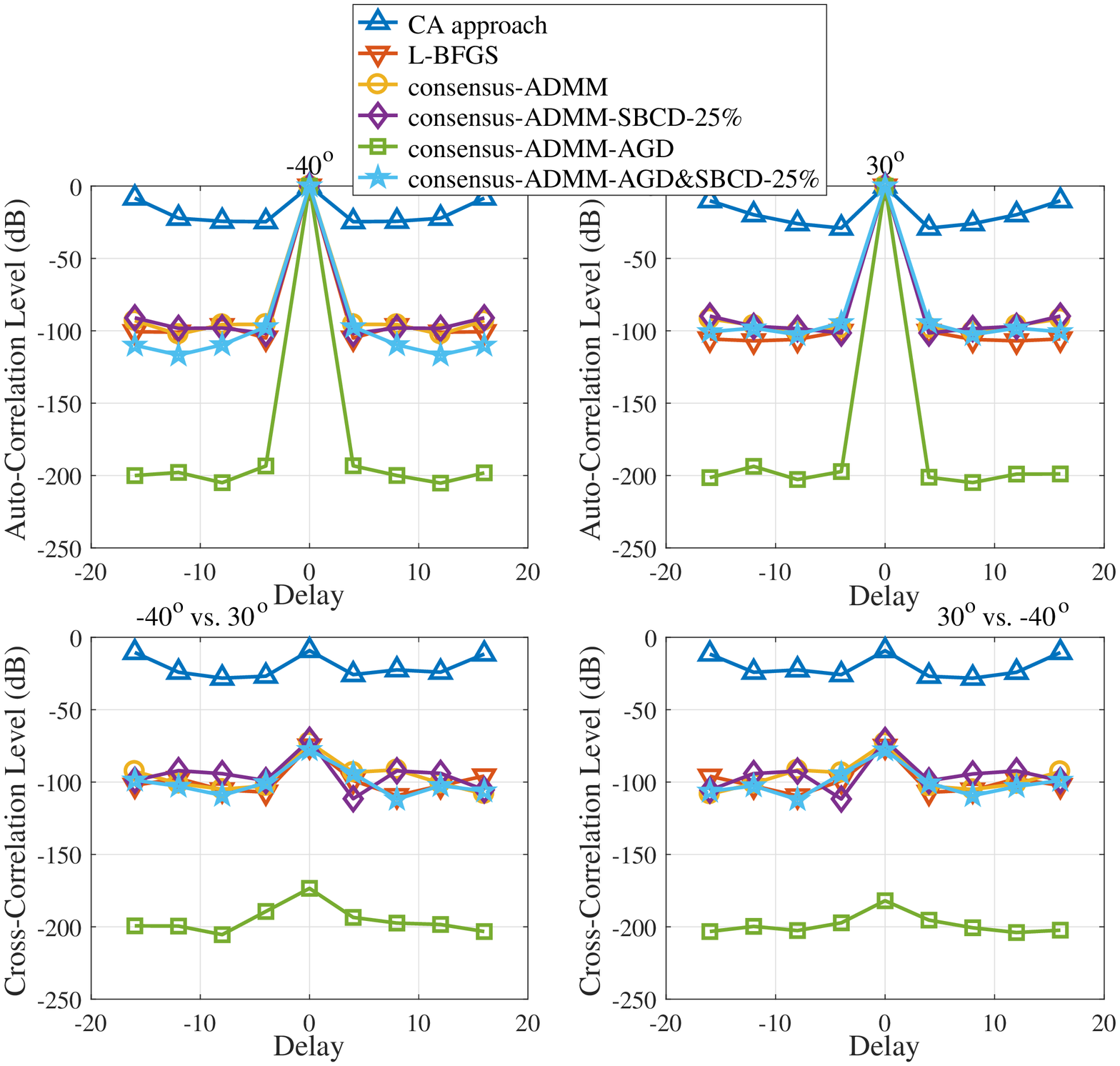}}
\caption{Comparison of correlation characteristics for interval $[0,16]$ with $M=8,N=128$.}
\label{correlation_b}
\end{figure}
\begin{figure}[htb]
  \centering
  \centerline{\includegraphics[scale=0.35]{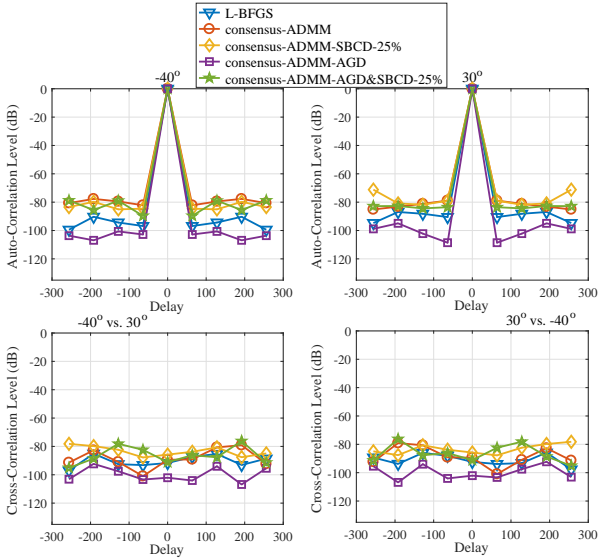}}
\caption{Comparison of correlation characteristics for interval $[0,256]$ with $M=128,N=1024$.}
\label{correlation_c}
\end{figure}
\begin{table*}[htbp]
\caption{Comparison of execution time (second per iteration).}
\label{table time}
\centering
\begin{tabular}{|c||c|c|c|c|c|c|}
\hline
\hline
~                         $(N,M,\mathcal{T})$
                         & CA           & D-ADMM
                         & L-BFGS               &   \tabincell{c}{consensus \\-ADMM} &  \tabincell{c}{consensus-ADMM\\  -SBCD-$25\%$} & \tabincell{c}{consensus\\  -ADMM-AGD} \\
\hline
                         (64,8,[0,16])
                         &  0.34s           & 316s
                         &  0.12s           & 0.08s  & 0.02s & 0.09s  \\
                         \hline
                         (128,8,[0,16])
                         &  0.67s           & 950s
                         &  0.14s           & 0.09s  & 0.02s & 0.10s   \\
                         \hline
                         (128,16,[0,32])
                         &  3.4s            & N/A
                         &  0.24s           & 0.26s & 0.08s &  0.28s \\
                         \hline
                         (\!1024,128,[0,256])
                         &  N/A             & N/A
                         &  5.1s            & 9.1s  &  2.3s & 11.7s   \\
\hline
\hline
\end{tabular}
\end{table*}

Table \ref{table time} shows averaged execution time (per iteration) of the proposed consensus-ADMM approach and three other state-of-the-art approaches \cite{Li_08}\cite{Wang_12}\cite{Cheng_17}. Here, N/A means that iteration operations cannot be finished in reasonable time.
 From the table, we can see that when $M$ and $N$ are small, the proposed consensus-ADMM algorithm and its variants (SBCD-$25\%$, AGD) have less execution time. When $N$ and $M$ are increased, for example, from (64,8) to (1024,128), execution time of the consensus-ADMM approach becomes comparable to L-BFGS approach. Combining these execution time per iteration with the curves in Figure \ref{convergence fig}, the total computational time can be computed easily. Moreover, we should note that in the proposed consensus-ADMM algorithm or its variants, parallel execution architecture plays an essential role leading to better implementation efficiency than state-of-the-art methods including L-BFGS, D-ADMM, and MM, which means that it is more suitable for large-scale applications from a practical viewpoint of implementation\footnotemark.

\footnotetext{In the real radar system, the algorithm is usually implemented using Field Programmable Gate Array (FPGA). This kind of integrated chip is very suitable for implementing an algorithm with a parallel structure.}

\section{Conclusion}
\label{sec:Conclusion}

In this paper, we focus on designing constant modulus probing waveforms with good correlation properties for the collocated MIMO radar system. By introducing auxiliary variables and exploiting the designing problem's inherent structures, we formulate a consensus-like optimization model. Then, the ADMM technique is customized to solve the corresponding non-convex problem approximately. We prove that the proposed ADMM approach is theoretically-guaranteed convergent if proper parameters are chosen. Simulation results show the effectiveness of the proposed consensus-ADMM algorithm and its variants, especially suitable for large-scale MIMO radar systems.

\section*{Acknowledgment}
The authors appreciate  Dr. Ian Clarkson, associate editor of this paper, and the anonymous reviewers who help the authors to elevate the quality of this paper.
\appendices
\section{Proof of Lemma 1}
In the following, we prove that both $\nabla h({\alpha,\mathbf{\Phi}})$ and $\nabla f_n(\mathbf{\Phi})$ are Lipschitz continuous via the definition of Lipschitz continuity. To state the proof clearly, we rewrite \eqref{quantities} and \eqref{Qv} in the following
\begin{equation}\label{def_quantities}
\begin{split}
&\mathbf{a}_{\theta,\theta}= {\rm vec}(\mathbf{a}_{\theta}\mathbf{a}_{\theta}^{H}),\ p = \sum_{\theta\in \Theta} \bar{P}_{\theta}^2,\\
&\mathbf{q} = -\sum_{\theta\in \Theta} \bar{P}_{\theta}\mathbf{a}_{\theta,\theta}, \ \mathbf{A} = \sum_{\theta\in \Theta} \mathbf{a}_{\theta,\theta} \mathbf{a}^H_{\theta,\theta}, \\
&\mathbf{v}(\alpha,\mathbf{\Phi}) = \left[
                       \begin{array}{l}
                        ~~~~~~~~~~\alpha\\
                        {\rm vec}\left(\mathbf{X}^H(\mathbf{\Phi})\mathbf{X}(\mathbf{\Phi})\right)
                       \end{array}
                       \right],\\
&\mathbf{Q} =  \left[
                       \begin{array}{ll}
                        p & \mathbf{q}^H\\
                        \mathbf{q} & \mathbf{A}
                       \end{array}
                       \right].
\end{split}
\end{equation}

To facilitate the subsequent derivations, we denote
\begin{equation}\label{z}
\begin{split}
\mathbf{z}(\mathbf{\Phi}) = {\rm vec}\left(\mathbf{X}^H(\mathbf{\Phi})\mathbf{X}(\mathbf{\Phi})\right).
\end{split}
\end{equation}
Then, $\mathbf{Qv}(\alpha,\mathbf{\Phi})$ can be expressed by
\begin{equation}\label{definintion Qv}
\begin{split}
\mathbf{Qv}(\alpha,\mathbf{\Phi}) =
\left[
\begin{array}{l}
                        p\alpha+\mathbf{q}^H\mathbf{z}(\mathbf{\Phi})  \\
                        \alpha \mathbf{q}+\mathbf{Az}(\mathbf{\Phi})
                       \end{array}
\right].
\end{split}
\end{equation}

First, we can obtain $\frac{\partial\mathbf{v}(\alpha,\mathbf{\Phi})}{\partial \alpha}=\left[1;\mathbf{0}\right]$. Plugging it and \eqref{definintion Qv} into \eqref{gradient alpha}, we can have
\begin{equation}\label{proof lipschitz alpha}
\begin{split}
&\frac{|\nabla_{\alpha} h({\alpha,\mathbf{\Phi}})-\nabla_{\alpha} h(\hat{\alpha},{\mathbf{\Phi}})|}{|\alpha-\hat{\alpha}|}\\
&=\frac{\big|2{\rm Re}\big(p\alpha+\mathbf{q}^H\mathbf{z}(\mathbf{\Phi}) -p\hat{\alpha}-\mathbf{q}^H\mathbf{z}(\mathbf{\Phi}) \big)\big|}{|\alpha-\hat{\alpha}|}= 2p,
\end{split}
\end{equation}
where $\alpha,\hat{\alpha}\in(0,\alpha_{\rm max}]$.
From \eqref{proof lipschitz alpha}, we can see that $\nabla_{\alpha} h({\alpha,\mathbf{\Phi}})$ is Lipschitz continuous with constant $L_{\alpha}\geq2p$.

Second, for $\nabla_{\mathbf{\Phi}} h({\alpha,\mathbf{\Phi}})$, we have the following derivations
\begin{equation}\label{proof lipschitz}
\begin{split}
& \frac{\|\nabla_{\mathbf{\Phi}} h({\alpha,\mathbf{\Phi}})-\nabla_{\mathbf{\Phi}} h(\alpha,\mathbf{\hat{\Phi}})\|^2_F}{\|{\mathbf{\Phi}}-\mathbf{\hat{\Phi}}\|^2_F}\\
 =& \frac{\displaystyle\sum_{i=1}^{N}\sum_{m=1}^M\left|\frac{\partial h(\alpha,\mathbf{\Phi})}{\partial \phi_{i,m}}-\frac{\partial h(\alpha,\mathbf{\hat{\Phi}})}{\partial \hat{\phi}_{i,m}}\right|^2}{\displaystyle\sum_{i=1}^{N}\sum_{m=1}^M|\phi_{i,m}-\hat{\phi}_{i,m}|^2}\\
 \leq& \max_{i,m}\left\{\left|\frac{\frac{\partial h(\alpha,\mathbf{\Phi})}{\partial \phi_{i,m}}-\frac{\partial h(\alpha,\mathbf{\hat{\Phi}})}{\partial \hat{\phi}_{i,m}}}{\phi_{i,m}-\hat{\phi}_{i,m}}\right|^2\right\}.
\end{split}
\end{equation}
According to Lagrange's mean value theorem, since $h(\alpha,\mathbf{\Phi})$ is continuous and differentiable, there exists some point $\bar{\phi}_{i,m}$ between $\phi_{i,m}$ and $\hat\phi_{i,m}$ which satisfies
\begin{equation}\label{mean value theorem}
\frac{\frac{\partial h(\alpha,\mathbf{\Phi})}{\partial \phi_{i,m}}-\frac{\partial h(\alpha,\mathbf{\hat{\Phi}})}{\partial \hat{\phi}_{i,m}}}{\phi_{i,m}-\hat{\phi}_{i,m}}= \frac{\partial^2 h(\alpha,\bar{\mathbf{\Phi}})}{\partial \bar{\phi}_{i,m}^2}.
\end{equation}
Combining \eqref{proof lipschitz} and \eqref{mean value theorem}, we obtain
\begin{equation}\label{lipschitz phi}
\begin{split}
\!\!\!\!\!\!\frac{\|\nabla\!_{\mathbf{\Phi}} h({\alpha,\!\mathbf{\Phi}})\!\!-\!\!\nabla\!_{\mathbf{\Phi}} h(\alpha,\!\mathbf{\hat{\Phi}})\|_F}{\|{\mathbf{\Phi}}-\mathbf{\hat{\Phi}}\|_F}\!\leq \max_{i,m}\!\left\{\!\left|\frac{\partial^2 h(\alpha,\bar{\mathbf{\Phi}}))}{\partial \bar{\phi}_{i,m}^2}\right|\!\right\}\!.
\end{split}
\end{equation}
From \eqref{gradient phi_0}, we can obtain
\[
\begin{split}
\hspace{-0.6cm}&\left|\frac{\partial^2 h(\alpha,\bar{\mathbf{\Phi}})}{\partial \bar{\phi}_{i,m}^2}\right|\\
\hspace{-0.5cm}\!=&\!\left|2{\rm Re}\!\left(\!\!\frac{\partial\mathbf{v}^H\!(\!\alpha,\bar{\mathbf{\Phi}}\!) }{\partial \bar{\phi}_{i,m}}\mathbf{Q}\frac{\partial\mathbf{v}(\!\alpha,\bar{\mathbf{\Phi}}\!) }{\partial \bar{\phi}_{i,m}}\!+\!\frac{\partial^2\mathbf{v}^H(\!\alpha,\bar{\mathbf{\Phi}}\!) }{\partial\bar{\phi}_{i,m}^2}\mathbf{Q}\mathbf{v}(\alpha,\bar{\mathbf{\Phi}})\!\!\right)\!\right|,
\end{split}
\]
which can be further derived as
\begin{equation}\label{max second derivative}
\begin{split}
\hspace{-0.6cm}&\left|\frac{\partial^2 h(\alpha,\bar{\mathbf{\Phi}})}{\partial \bar{\phi}_{i,m}^2}\right|\\
\hspace{-0.5cm}\leq&2\left|\frac{\partial\mathbf{v}^H(\!\alpha,\bar{\mathbf{\Phi}}\!) }{\partial \bar{\phi}_{i,m}}\mathbf{Q}\frac{\partial\mathbf{v}(\!\alpha,\bar{\mathbf{\Phi}}\!) }{\partial \bar{\phi}_{i,m}}\right|+ 2\left|\frac{\partial^2\mathbf{v}^H\!(\!\alpha,\bar{\mathbf{\Phi}}) }{\partial\bar{\phi}_{i,m}^2}\mathbf{Q}\mathbf{v}(\!\alpha,\bar{\mathbf{\Phi}}\!)\right|.
\end{split}
\end{equation}
Since $\frac{\partial\mathbf{v}(\alpha,\bar{\mathbf{\Phi}})}{\partial{\bar\phi_{i,m}}} = \left[0; \frac{\partial\mathbf{z}(\bar{\mathbf{\Phi}}) }{\partial \bar{\phi}_{i,m}}\right]$, the first term in \eqref{max second derivative}'s right side can be derived as (see \eqref{def_quantities})
\begin{equation}\label{AQ}
\begin{split}
\left|\!\frac{\partial\mathbf{v}^H{(\alpha,\bar{\mathbf{\Phi}})} }{\partial \bar{\phi}_{i,m}}\mathbf{Q}\frac{\partial\mathbf{v}{(\alpha,\bar{\mathbf{\Phi}})} }{\partial \bar{\phi}_{i,m}}\right|=\left|\!\frac{\partial\mathbf{z}^H{(\bar{\mathbf{\Phi}})} }{\partial \bar{\phi}_{i,m}}\mathbf{A}\frac{\partial\mathbf{z}{(\bar{\mathbf{\Phi}})} }{\partial \bar{\phi}_{i,m}}\right|.
\end{split}
\end{equation}
\begin{figure*}

\begin{equation}\label{second partial derivative}
\begin{split}
 \frac{\partial^2 \mathbf{X}^{\!H}(\mathbf{\Phi})\mathbf{X}(\mathbf{\Phi})}{\partial\phi_{i,m}^2}\!\!=\!\!\!
\left[\!\!
  \begin{array}{ccccccc}
    ~ &  ~& ~ & \!-e^{j(\phi_{i,m}\!-\phi_{i,1})} & ~ & ~ & ~ \\
    ~ &  \! \! \mathbf{\scalebox{3.5}0}\!\! & ~& \vdots & ~  & \!\! \mathbf{\scalebox{3.5}0}\!\!  & ~ \\
   ~  &~ &~ & \!-e^{j(\phi_{i,m}\!-\phi_{i,m\!-\!1})}&~ &~  \\
    \!\!-e^{j(\phi_{i,1}\!-\!\phi_{i,m})} &\!\!\cdots\!\! &\!-e^{j(\phi_{i,m\!-\!1}\!-\phi_{i,m})} & 0 & \!-e^{j(\phi_{i,m+\!1}\!-\phi_{i,m})}&  \!\!\cdots\!\! & \!-\!e^{j(\phi_{i,M}\!-\phi_{i,m})}\!\!  \\
    ~ & ~ &~ &\!-e^{j(\phi_{i,m}\!-\phi_{i,m+\!1})}  &~ & ~& ~  \\
    ~ &  \!\!\mathbf{\scalebox{3.5}0} \!\! & ~ & \vdots & ~ & \!\! \mathbf{\scalebox{3.5}0} \!\! & ~  \\
    ~ & ~ &~ & \!-e^{j(\phi_{i,m}\!-\phi_{i,M})} &~ & ~& ~  \\
  \end{array}
\!\!\right].
\end{split}
\end{equation}
\hrulefill
\vspace*{4pt}
\end{figure*}
Since $\mathbf{z}(\bar{\mathbf{\Phi}}) = {\rm vec}(\mathbf{X(\bar{\mathbf{\Phi}})}^H\mathbf{X}(\bar{\mathbf{\Phi}}))$ and either $\frac{\partial \mathbf{X}^{\!H}(\bar{\mathbf{\Phi}})\mathbf{X}({\bar{\mathbf{\Phi}}})}{\partial{\bar{\phi}_{i,m}}}$ or $\frac{\partial^2 \mathbf{X}^{\!H}(\bar{\mathbf{\Phi}})\mathbf{X}({\bar{\mathbf{\Phi}}})}{\partial{\bar{\phi}_{i,m}}^2}$ is an $M$-by-$M$ matrix (see them in \eqref{partial derivative} and \eqref{second partial derivative} respectively), we can see that $\frac{\partial\mathbf{z}(\mathbf{\Phi})}{\partial \phi_{i,m}}$ and $\frac{\partial^2\mathbf{z}(\mathbf{\Phi})}{\partial \phi_{i,m}^2}$ has $2(M-1)$ nonzero constant modulus elements respectively.
Since the maximum modulus of elements (MME) in $\bf{A}$ is  $|\Theta|$, \eqref{AQ} can be further derived as
\begin{equation}\label{bound first}
\begin{split}
\left|\!\frac{\partial\mathbf{v}^H{(\alpha,\bar{\mathbf{\Phi}})} }{\partial \bar{\phi}_{i,m}}\mathbf{Q}\frac{\partial\mathbf{v}{(\alpha,\bar{\mathbf{\Phi}})} }{\partial \bar{\phi}_{i,m}}\right| \leq{4(M-1)^2|\Theta|}.
\end{split}
\end{equation}

Similarly, since $\frac{\partial^2\mathbf{v} {(\alpha,\bar{\mathbf{\Phi}})}}{\partial \phi_{i,m}^2} = \left[0;\frac{\partial^2\mathbf{z} {(\bar{\mathbf{\Phi}})} }{\partial \bar{\phi}_{i,m}^2}\right]$  {and \eqref{definintion Qv} holds, the second term in \eqref{max second derivative}'s right side can be denoted as}
\begin{equation}\label{vQv}
\begin{split}
\left|\!\frac{\partial^2\mathbf{v}^H {(\alpha,\bar{\mathbf{\Phi}})} }{\partial\bar{\phi}_{i,m}^2}\mathbf{Q}\mathbf{v} {(\alpha,\bar{\mathbf{\Phi}})}\right| \!\!= \!\!\left|\!\frac{\partial^2\mathbf{z}^H {(\bar{\mathbf{\Phi}})} }{\partial\bar{\phi}_{i,m}^2}
\left[ \alpha\mathbf{q} + \mathbf{Az {(\bar{\mathbf{\Phi}})}}
\right]\right|.
\end{split}
\end{equation}
 {Since MME in $\mathbf{a}_{\theta,\theta}$ is 1, we can see that MMEs in $\mathbf{q}$ and $\mathbf{A} $ are $\bar{P}_{\rm max}|\Theta|$ and $|\Theta|$ respectively, where $\bar{P}_{\rm max}=\max\limits_{\theta\in\Theta}\{\bar{P}_{\theta}\}$, i.e.,  $\bar{P}_{\rm max}$ is the maximum element of the desired beampattern.
Since $\mathbf{X}(\mathbf{\Phi})\in\mathds{C}^{N\times M}$ and $x_{i,m}=e^{j\phi_{i,m}}$, MME in $\mathbf{z}$ is $N$.}
Then, \eqref{vQv} can be derived as
\begin{equation}\label{bound second}
\begin{split}
\hspace{-0.3cm}\left|\!\frac{\partial^2\mathbf{v}^H {(\alpha,\bar{\mathbf{\Phi}})} }{\partial\bar{\phi}_{i,m}^2}\mathbf{Q}\mathbf{v} {(\alpha,\bar{\mathbf{\Phi}})}\right| \!\!\leq\! \! {2(M\!-\!1)(\alpha\bar{P}_{\max}\!+\!M^2N)|\Theta|}.
\end{split}
\end{equation}
 {Plugging \eqref{bound first} and \eqref{bound second} into \eqref{max second derivative}, we have}
\begin{equation}\label{bound sec h}
\begin{split}
\hspace{-0.3cm}\left|\frac{\partial^2 h(\alpha, {\bar{\mathbf{\Phi}}})}{\partial \bar{\phi}_{i,m}^2}\right|
\leq\! { 4(\!M\!-\!1\!)(\!\alpha_{\rm max}\bar{P}_{\rm max}\!+\!M^2N\!+\!2M\!-\!2\!)|\Theta|}.
\end{split}
\end{equation}
 {Moreover, plugging \eqref{bound sec h} into \eqref{lipschitz phi}, we can obtain
\begin{equation*}\label{Lip defi}
\begin{split}
&\frac{\|\nabla\!_{\mathbf{\Phi}} h({\alpha,\!\mathbf{\Phi}})\!\!-\!\!\nabla\!_{\mathbf{\Phi}} h(\alpha,\!\mathbf{\hat{\Phi}})\|_F}{\|{\mathbf{\Phi}}-\mathbf{\hat{\Phi}}\|_F}\\
\leq& 4(M\!-\!1)(\alpha_{\rm max}\bar{P}_{\rm max}\!+\!M^2N\!+\!2M\!-\!2)|\Theta|.
\end{split}
\end{equation*}
Therefore,$\nabla_{\mathbf{\Phi}} h(\alpha,\mathbf{\Phi})$ is Lipschitz continuous with the constant $L\geq4(M\!-\!1)(\alpha_{\rm max}\bar{P}_{\rm max}\!+\!M^2N\!+\!2M\!-\!2)|\Theta|$.}

Third, for $\nabla f_n(\mathbf{\Phi})$, there exists
\begin{equation}\label{lipschitz phi n}
\begin{split}
\frac{\|\nabla f_n({\mathbf{\Phi}})-\nabla f_n(\mathbf{\hat{\Phi}})\|_F}{\|{\mathbf{\Phi}}\!-\!\mathbf{\hat{\Phi}}\|_F} \!\leq \max_{i,m}\left\{\left|\frac{\partial^2 f_n( {\bar{\mathbf{\Phi}}})}{\partial\bar{\phi}_{i,m}^2} \right|\right\}.
\end{split}
\end{equation}
For $\frac{\partial^2 f_n( {\bar{\mathbf{\Phi}}})}{\partial\bar{\phi}_{i,m}^2} $, we have
\begin{equation}\label{second derivative phi n}
\begin{split}
&\left|\frac{\partial^2 f_n( {\bar{\mathbf{\Phi}}})}{\partial\bar{\phi}_{i,m}^2} \right|\\
\hspace{-0.3cm}=&\!\left|2{\rm Re}\!\left(\!{\rm Tr}\bigg(\!
\frac{\partial\mathbf{B}_n^H( {\bar{\mathbf{\Phi}}}) }{\partial \bar{\phi}_{i,m}}\!\frac{\partial\mathbf{B}_n( {\bar{\mathbf{\Phi}}}) }{\partial \bar{\phi}_{i,m}}\!+\!
\mathbf{B}_n^H( {\bar{\mathbf{\Phi}}})\frac{\partial^2\mathbf{B}_n( {\bar{\mathbf{\Phi}}}) }{\partial \bar{\phi}_{i,m}^2}  \!\bigg)\!\!\right)\right|\\
\hspace{-0.3cm}\leq& 2\left|{\rm Tr}\bigg(\!\!
\frac{\partial\mathbf{B}_n^H( {\bar{\mathbf{\Phi}}}) }{\partial \bar{\phi}_{i,m}}\!\frac{\partial\mathbf{B}_n( {\bar{\mathbf{\Phi}}}) }{\partial \bar{\phi}_{i,m}}\!\!\bigg)\!\right|\!\!+\!2\!\left|{\rm Tr}\bigg(\!\!\mathbf{B}_n^H( {\bar{\mathbf{\Phi}}})\frac{\partial^2\mathbf{B}_n( {\bar{\mathbf{\Phi}}}) }{\partial \bar{\phi}_{i,m}^2} \! \!\bigg)\!\right|.
\end{split}
\end{equation}

\begin{figure*}
\normalsize
 {
\begin{equation}\label{XSX}
\begin{split}
  &\mathbf{X(\Phi)}^H\mathbf{S}_n\mathbf{X(\Phi)}\!\! =\!\!\!
  \begin{bmatrix}
     \displaystyle\sum_{i=1}^{N-n}e^{j\left(\phi_{i+n,1}-\phi_{i,1}\right)}  & \displaystyle\sum_{i=1}^{N-n}e^{j\left(\phi_{i+n,2}-\phi_{i,1}\right)} & \dotsb &\displaystyle\sum_{i=1}^{N-n}e^{j\left(\phi_{i+n,m}-\phi_{i,1}\right)} & \dotsb &  \displaystyle\sum_{i=1}^{N-n}e^{j\left(\phi_{i+n,M}-\phi_{i,1}\right)} \\
     \displaystyle\sum_{i=1}^{N-n}e^{j\left(\phi_{{i+n,1}}-\phi_{i,2}\right)} & \displaystyle\sum_{i=1}^{N-n}e^{j\left(\phi_{{i+n,2}}-\phi_{i,2}\right)} & \dotsb &\displaystyle\sum_{i=1}^{N-n}e^{j\left(\phi_{i+n,m}-\phi_{i,2}\right)}& \dotsb &  \displaystyle\sum_{i=1}^{N-n}e^{j\left(\phi_{{i+n,M}}-\phi_{i,2}\right)}\\
      \vdots & \vdots & \vdots & \vdots & \vdots & \vdots \\
     \displaystyle\sum\limits_{i=1}^{N-n}e^{j\left(\phi_{i+n,1}-\phi_{i,M}\right)} & \displaystyle\sum_{i=1}^{N-n}e^{j\left(\phi_{i+n,2}-\phi_{i,M}\right)} & \dotsb &\displaystyle\sum_{i=1}^{N-n} e^{j(\phi_{i+n,m}-\phi_{i,M})} & \dotsb &  \displaystyle\sum_{i=1}^{N-n}e^{j\left(\phi_{{i+n,M}}-\phi_{i,M}\right)}
  \end{bmatrix}.
\end{split}
\end{equation}}
\hrulefill
\vspace*{4pt}
\end{figure*}
 {The elements in  $\mathbf{B}_n({\bar{\mathbf{\Phi}}})$ are $P_{\theta_i,\theta_j,n}$ or $0$, where $P_{\theta_i,\theta_j,n} = \mathbf{a}_{\theta_i}^{H}\mathbf{X}^H\mathbf{S}_n\mathbf{X}\mathbf{a}_{\theta_j}$. Since $\mathbf{a}_\theta = \left[1,e^{j\pi \sin\theta},\cdots,e^{j\pi(M-1) \sin\theta}\right]^T$ and \eqref{XSX} holds, we can get
\begin{equation}\label{pn}
\begin{split}
P_{\theta_i,\theta_j,n}\!=\!\!\!\displaystyle\sum_{q=1}^M\!\!\sum_{m=1}^M\!\!\sum_{i=1}^{N-n}\!e^{j\pi (q-\!1)(\sin\theta_i+\sin\theta_j\!)}\!e^{j(\phi_{{i+n,q}}\!-\!\phi_{i,m})}.
\end{split}
\end{equation}
From \eqref{pn}, it is easy to see that MME in all the $P_{\theta_i,\theta_j,n}$ is no larger than $M^2N$.
Therefore, MME in  $\mathbf{B}_n( {\bar{\mathbf{\Phi}}})$ is no larger than $w_{\rm c}M^2N$, where $w_{\rm c}= \max\{w_{\rm ac},w_{\rm cc}\}$.}

 {Next, we consider MMEs in $\frac{\partial\mathbf{B}_n({\bar{\mathbf{\Phi}}}) }{\partial \bar{\phi}_{i,m}}$ and $\frac{\partial^2\mathbf{B}_n({\bar{\mathbf{\Phi}}}) }{\partial \bar{\phi}_{i,m}^2}$. From \eqref{XSX}, we can get
\[
\begin{split}
&\mathbf{a}_{\theta_i}^{H}\!\frac{\partial \mathbf{X}^{H}(\mathbf{\Phi})\mathbf{S}_n\mathbf{X}(\mathbf{\Phi})}{\partial\phi_{i,m}}\mathbf{a}_{\theta_j}\\
=&\displaystyle\sum_{t=1\atop{t\neq m}}^M-je^{j\pi (t-1)(\sin\theta_i+\sin\theta_j)}e^{j(\phi_{{i+n,t}}-\phi_{i,m})}\\
&+\sum_{q=1}^Mje^{j(\phi_{{i+n,m}}-\phi_{i,q})}e^{j\pi (m-1)(\sin\theta_i+\sin\theta_j)},\\
&\mathbf{a}_{\theta_i}^{H}\frac{\partial^2 \mathbf{X}^{H}(\mathbf{\Phi})\mathbf{S}_n\mathbf{X}(\mathbf{\Phi})}{\partial\phi_{i,m}^2}\mathbf{a}_{\theta_j}\\
=&\displaystyle\sum_{t=1\atop{t\neq m}}^M-e^{j\pi (t-1)(\sin\theta_i+\sin\theta_j)}e^{j(\phi_{{i+n,t}}-\phi_{i,m})}\\
&+\sum_{q=1}^M-e^{j(\phi_{{i+n,m}}-\phi_{i,q})}e^{j\pi (m-1)(\sin\theta_i+\sin\theta_j)}.
\end{split}
\]
which indicate
\[
 \begin{split}
 &\left|\mathbf{a}_{\theta_i}^{H}\frac{\partial \mathbf{X}^{\!H}(\mathbf{\Phi})\mathbf{S}_n\mathbf{X}(\mathbf{\Phi})}{\partial\phi_{i,m}}\mathbf{a}_{\theta_j}\right|\leq 2M-1,\\
 &\left|\mathbf{a}_{\theta_i}^{H}\frac{\partial^2 \mathbf{X}^{\!H}(\mathbf{\Phi})\mathbf{S}_n\mathbf{X}(\mathbf{\Phi})}{\partial\phi_{i,m}^2}\mathbf{a}_{\theta_j}\right|\leq 2M-1.
 \end{split}
 \]
Therefore, we can conclude that MME in either $\frac{\partial\mathbf{B}_n({\bar{\mathbf{\Phi}}}) }{\partial \bar{\phi}_{i,m}}$ or $\frac{\partial^2\mathbf{B}_n({\bar{\mathbf{\Phi}}}) }{\partial \bar{\phi}_{i,m}^2}$  is no larger than $w_{\rm c}(2M-1)$.}
 {Plugging these results into \eqref{second derivative phi n}, we can obtain
\begin{equation}\label{bound phi n}
\begin{split}
\hspace{-0.3cm}\left|\frac{\partial^2 f_n({\bar {\mathbf{\Phi}}})}{\partial\bar{\phi}_{i,m}^2} \right| \! \leq 2w_{\rm c}^2(2M\!-\!1)(\!M^2N\!+\!2M\!-\!1)K^2.
\end{split}
\end{equation}
where $K$ is the number of considered spacial directions.
Plugging \eqref{bound phi n} into \eqref{lipschitz phi n}, we can obtain
\begin{equation}\label{lip defi2}
\begin{split}
\frac{\|\!\nabla f_n({\mathbf{\Phi}})\!\!-\!\!\nabla f_n(\mathbf{\hat{\Phi}})\!\|_{\!F}}{\|{\mathbf{\Phi}}\!-\!\mathbf{\hat{\Phi}}\|_{F}} \!\!\leq\!\! 2w_{\rm c}^2(2M\!\!-\!\!1)(\!M^2N\!\!+\!\!2M\!\!-\!\!1\!)K^2.
\end{split}
\end{equation}
Therefore, we can see that functions $\nabla f_n(\mathbf{\Phi}), n\in\mathcal{T}$ are Lipschitz continuous with the constant $L_n\! \!\geq\!\!  2w_{\rm c}^2(\!2M\!-\!1\!)(\!M^2N\!\!+\!\!2M\!\!-\!\!1)K^2$.}
This concludes the proof. $\hfill\blacksquare$

\section{Proof of Lemmas \ref{upperbound function}--\ref{lemma lower bound}}
\begin{lemma}\label{upperbound function}
 The following two inequalities exist
\begin{equation}\label{diff U}
\begin{split}
\!&\mathcal{U}(\alpha,\!\mathbf{\Phi},\! \{\mathbf{\Phi}_n^k,\!\mathbf{\Lambda}_n^k,\!n\in\mathcal{T}\})\!-\!\mathcal{L}(\alpha,\!\mathbf{\Phi},\!\{\mathbf{\Phi}_n^k,\!\mathbf{\Lambda}_n^k,n\in\mathcal{T}\}) \\
\!\leq& \!\ 2L_{\alpha}|\alpha-\alpha^k|^2+2L\|{\mathbf{\Phi}}-{\mathbf{\Phi}^{k}}\|_F^2. \end{split}
\end{equation}
\begin{equation}\label{diff Un}
\begin{split}
&\mathcal{U}_n\left(\mathbf{\Phi}^{k+1}, \mathbf{\Phi}_n,\mathbf{\Lambda}_n^k\right)-\mathcal{L}_n(\mathbf{\Phi}^{k+1}, \mathbf{\Phi}_n, \mathbf{\Lambda}_n^k)\\
&\hspace{2.5cm}\leq 2L_n\|{\mathbf{\Phi}_n}-{\mathbf{\Phi}^{k+1}}\|_F^2, \ \ \forall n\in \mathcal{T}.
\end{split}
\end{equation}
\end{lemma}
\begin{proof} For \eqref{diff U}, its left side can be derived as
\begin{equation}\label{difference phi}
\begin{split}
\!\!\!\! &\!\mathcal{U}(\alpha,\mathbf{\Phi},\! \{\mathbf{\Phi}_n^k,\mathbf{\Lambda}_n^k, n\in\mathcal{T}\}\!)\!-\!\mathcal{L}(\alpha,\mathbf{\Phi}\!,\!\{\mathbf{\Phi}_n^k\!,\mathbf{\Lambda}_n^k, n\in\mathcal{T}\}\!)\\
\!\!=& h(\alpha^k,\mathbf{\Phi}^k)\!-\!h(\alpha,\mathbf{\Phi})\!+\!\langle \nabla_{\mathbf\Phi} h(\alpha^k, {\mathbf{\Phi}}^{k}),\mathbf\Phi\!-\!\mathbf\Phi^k\rangle \\
\!\!+& \langle \nabla_{\!\alpha} h(\alpha^k,\! {\mathbf{\Phi}}^{k}),\!\alpha\!\!-\!\alpha^k\rangle\!\!+\!\!\frac{L}{2}\|{\mathbf{\Phi}}\!-\!\!{\mathbf{\Phi}}^{k}\|_F^2\!+\!\!\frac{L_{\alpha}}{2}|\alpha\!\!-\!\alpha^k|^2.
\end{split}
\end{equation}
 {In Lemma \ref{Lipschtiz continuous}, we show that gradient $\nabla h(\alpha,\mathbf{\Phi})$ is Lipschitz continuous.
According to the Decent Lemma \cite{Bertsekas_99}, we have
\begin{equation}\label{ineq h}
\begin{split}
h(\alpha,\mathbf{\Phi})&\leq h(\alpha^k,\mathbf{\Phi}^k)  +\ \langle \nabla_{\mathbf\Phi} h(\alpha^k,{\mathbf{\Phi}^k}),{\mathbf{\Phi}}-{\mathbf{\Phi}^{k}}\rangle\\
&+\langle \nabla_{\alpha} h(\alpha^k,{\mathbf{\Phi}}^k),\alpha\!-\!\alpha^k\rangle+\frac{L}{2}\|{\mathbf{\Phi}}-{\mathbf{\Phi}^{k}}\|_F^2\\
&+\frac{L_{\alpha}}{2}|\alpha\!-\!\alpha^k|^2,
\end{split}
\end{equation}
which can be further derived to the following inequality}
\begin{equation}\label{lipschitz property}
\begin{split}
&h(\alpha^k,\mathbf{\Phi}^k)-h(\alpha,\mathbf{\Phi})   \\
\leq &\ \langle \nabla_{\mathbf\Phi} h(\alpha^k,{\mathbf{\Phi}}),{\mathbf{\Phi}^{k}}-{\mathbf{\Phi}}\rangle+\frac{L}{2}\|{\mathbf{\Phi}}^{k}-{\mathbf{\Phi}}\|_F^2\\
&+\langle \nabla_{\alpha} h(\alpha,{\mathbf{\Phi}}^k),\alpha^k\!-\!\alpha\rangle+\frac{L_{\alpha}}{2}|\alpha^k\!-\!\alpha|^2.
\end{split}
\end{equation}
Plugging \eqref{lipschitz property} into \eqref{difference phi}, we have the following derivations
\begin{equation}\label{difference phi 1}
\begin{split}
\!\!\!\! &\!\mathcal{U}(\alpha,\mathbf{\Phi}, \{\mathbf{\Phi}_n^k,\mathbf{\Lambda}_n^k, n\in\mathcal{T}\})\!\!-\!\mathcal{L}(\alpha,\mathbf{\Phi}\!,\!\{\mathbf{\Phi}_n^k,\mathbf{\Lambda}_n^k, n\in\mathcal{T}\})\\
\!\!\!\!\leq&\langle \nabla_{\mathbf\Phi} h(\alpha^k\!,\mathbf{\Phi}^k)\!- \!\nabla_{\mathbf\Phi} h(\alpha^k\!,\!\mathbf{\Phi}),{\mathbf{\Phi}}^k\!-\!{\mathbf{\Phi}}\rangle \!+\!L\|{\mathbf{\Phi}}\!-\!{\mathbf{\Phi}^k}\|_F^2\\
\!\!\!\!+&\langle \nabla_{\alpha} h(\alpha^k,\mathbf{\Phi}^k)\!-\!\nabla_{\alpha} h(\alpha,\mathbf{\Phi}^k),{\alpha^k-\alpha}\rangle\!+\!L_{\alpha}|\alpha\!-\!\alpha^k|^2.
\end{split}
\end{equation}
Furthermore, since
\[
 \begin{split}
   &\langle \nabla_{\mathbf\Phi} h(\alpha^k\!,\mathbf{\Phi}^k)\!- \!\nabla_{\mathbf\Phi} h(\alpha^k\,\ \mathbf{\Phi}),{\mathbf{\Phi}}^k\!-\!{\mathbf{\Phi}}\rangle
   \leq L\lVert\mathbf{\Phi}^k\!-\!\mathbf{\Phi}\rVert^2_F,\\
   &\langle \nabla_{\alpha} h(\alpha^k,\mathbf{\Phi}^k)\!-\!\nabla_{\alpha} h(\alpha,\mathbf{\Phi}^k),{\alpha^k-\alpha}\rangle \leq L_{\alpha} |\alpha^k-\alpha|^2,
 \end{split}
\]
we can get \eqref{diff U} from \eqref{difference phi 1}.

Exploiting the Lipschitz continuous property of $\nabla f_n(\mathbf{\Phi})$ and similar derivations, \eqref{diff Un} can also be obtained.
$\hfill\blacksquare$
\end{proof}

\begin{lemma}\label{lemma difference lambda}
in each iteration, $\forall~n\in \mathcal{T} $, $\|{\bf\Lambda}_n^{k+1}-{\bf\Lambda}_n^k\|_F^2$ can be bounded, i.e.,
\begin{equation}\label{difference lambda}
\begin{split}
\hspace{-0.3cm}\|{\mathbf{\Lambda}_n^{k+1}}\!\!-\!\!{\mathbf{\Lambda}_n^k}\|_F^2 \!\!\leq\! 2L_n^2\!\left(2\|{\mathbf{\Phi}_n^{k+1}}\!\!\!-\!{\mathbf{\Phi}_n^k}\|_F^2\!+\!3\|{\mathbf{\Phi}^{k+1}}\!\!-\!\!{\mathbf{\Phi}^{k}}\|_F^2\!\right).
\end{split}
\end{equation}
\end{lemma}

\begin{proof}
The optimal solutions of the problems \eqref{relax ADMM b} can be obtained by solving $\nabla_{\mathbf{\Phi}_n}\mathcal{U}_n\left(\mathbf{\Phi}^{k+1},\mathbf{\Phi}_n,\mathbf{\Lambda}_n^k\right)=0$, $\forall~n\in \mathcal{T}$, i.e.,
\begin{equation}\label{solve pdmm b}
\begin{split}
\nabla f_n({\mathbf{\Phi}}^{k+1})+ \mathbf{\Lambda}_n^k+(\rho_n+L_n)(\mathbf{\Phi}_n^{k+1}-\mathbf{\Phi}^{k+1})=0.
\end{split}
\end{equation}
Combining \eqref{solve pdmm b} and \eqref{relax ADMM c}, we can obtain
\begin{equation}\label{lambda update}
\begin{split}
\!\!\!\!\!\mathbf{\Lambda}_n^{k+1}=-\nabla f_n({\mathbf{\Phi}}^{k+1})- L_n(\mathbf{\Phi}_n^{k+1}-\mathbf{\Phi}^{k+1}).
\end{split}
\end{equation}
Plugging \eqref{lambda update} into  $\|\mathbf{\Lambda}_n^{k+1}-\mathbf{\Lambda}_n^{k}\|_F^2$, we have the following derivations
\[
\begin{split}
\!\!\!&\|\mathbf{\Lambda}_n^{k+1}-\mathbf{\Lambda}_n^{k}\|_F^2\\
\!\!\!=&\|\nabla f_n(\!{\mathbf{\Phi}^{k+1}}\!)\!-\!\nabla f_n({\mathbf{\Phi}^{k}})\!+\!L_n(\mathbf{\Phi}_n^{k+1} \!-\!\mathbf{\Phi}^{k+1}\!-\!\mathbf{\Phi}_n^{k}\!+\!\mathbf{\Phi}^{k})\|_F^2\\
\!\!\!{\leq}&2\|\!\nabla\! f_n(\!{\mathbf{\Phi}^{k+1}}\!)\!-\!\!\nabla \!f_n(\!{\mathbf{\Phi}^{k}})\!\|_F^2\!+\!\!2L_n^2\|\!\mathbf{\Phi}_n^{k+1} \!-\!\mathbf{\Phi}_n^{k}\!-\!\mathbf{\Phi}^{k+1}\!+\!\mathbf{\Phi}^{k}\!\|_F^2\\
\!\!\!{\leq}&\ 2L_n^2(2\|{\mathbf{\Phi}_n^{k+1}}-{\mathbf{\Phi}_n^{k}}\|_F^2+3\|{\mathbf{\Phi}^{k+1}}-{\mathbf{\Phi}^{k}}\|_F^2),
\end{split}
\]
where the second inequality comes from Lemma 1. This completes the proof. $\hfill\blacksquare$
\end{proof}

\begin{lemma}\label{successive difference}
let $\bar{c}_{n}\!=\!\rho_n^3\!-\!7\rho_n^2L_n\!-\! {8\rho_nL_n^2\!-\!32L_n^3}$ and $\tilde{c}_{n}\!=\! \rho_n^3\!-\!12\rho_nL_n^2\!-\!48L_n^3$. If $\bar{c}_{n}>0$ and $\tilde{c}_{n}>0$, in each consensus-ADMM iteration, the augmented Lagrangian function $\mathcal{L}(\cdot)$ {\it decreases sufficiently}, i.e.,
\begin{equation}\label{succ-diff phi}
\begin{split}
\!\!\!\!\!& \mathcal{L}(\alpha^{k},\mathbf{\Phi}^k, \{\mathbf{\Phi}_n^k,\mathbf{\Lambda}_n^k, n\in\mathcal{T}\})\\
&\hspace{1.5cm}-\mathcal{L}(\alpha^{k+1},\mathbf{\Phi}^{k+1},\{\mathbf{\Phi}_n^{k+1},\mathbf{\Lambda}_n^{k+1}, n\in\mathcal{T}\})\\
\!\!\!\!\!\geq& \sum_{n\in \mathcal{T}}\frac{1}{2\rho_n^2}\left( \bar{c}_{n}\|{\mathbf{\Phi}_n^{k+1}}\!-\!{\mathbf{\Phi}_n^k}\|_F^2\!+\!\tilde{c}_{n} \|{\mathbf{\Phi}^{k+1}}\!-\!{\mathbf{\Phi}^{k}}\|_F^2 \right) \\
\!\!\!\!\!&+\frac{L}{2}\|{\mathbf{\Phi}^{k+1}}-{\mathbf{\Phi}^{k}}\|_F^2 +\frac{L_{\alpha}}{2}|\alpha^{k+1}-\alpha^k|^2.
\end{split}
\end{equation}
\end{lemma}
\begin{proof}
Define the following quantities
\[
  \begin{split}
    &\Delta_{\alpha,\mathbf{\Phi}}^k\!=\!\mathcal{L}(\alpha^{k}\!,\!\mathbf{\Phi}^{k}\!, \{\mathbf{\Phi}_n^k,\mathbf{\Lambda}_n^k, n\in\mathcal{T}\}) \\
    &\hspace{2.6cm}-\!\mathcal{L}(\alpha^{k+1}\!,\mathbf{\Phi}^{k+1}\!,\!\{\mathbf{\Phi}_n^{k},\!\mathbf{\Lambda}_n^{k},n\in\mathcal{T}\}),\\
    &\Delta_{\mathbf{\Phi}_n}^k\!\!=\!\mathcal{L}_n(\mathbf{\Phi}^{k+1}\!, \mathbf{\Phi}_n^k,\mathbf{\Lambda}_n^k)-\!\mathcal{L}_n(\mathbf{\Phi}^{k+1}\!,\!\mathbf{\Phi}_n^{k+1}\!,\!\mathbf{\Lambda}_n^{k}),\\
    &\Delta_{\mathbf{\Lambda}_n}^k \!\!=\!\mathcal{L}_n(\mathbf{\Phi}^{k\!+\!1}\!, \!\mathbf{\Phi}_n^{k+1}\!,\!\mathbf{\Lambda}_n^k)\!-\!\mathcal{L}_n(\mathbf{\Phi}^{k+1}\!,\mathbf{\Phi}_n^{k+1}\!,\!\mathbf{\Lambda}_n^{k+1}).\\
  \end{split}
\]
Then, we can get
\begin{equation}\label{diff L}
  \begin{split}
  &\mathcal{L}(\alpha^{k},\mathbf{\Phi}^k, \{\mathbf{\Phi}_n^k,\mathbf{\Lambda}_n^k,n\in\mathcal{T}\}) \\
  &\hspace{1.5cm}-\!\mathcal{L}(\alpha^{k+1},\!\mathbf{\Phi}^{k+1},\{\mathbf{\Phi}_n^{k+1},\mathbf{\Lambda}_n^{k+1},n\in\mathcal{T}\}) \\
= & \Delta_{\alpha,\mathbf{\Phi}}^k + \displaystyle\sum_{n\in\mathcal{T}} \left(\Delta_{\mathbf{\Lambda}_n}^k + \Delta_{\mathbf{\Phi}_n}^k\right).
  \end{split}
\end{equation}

For $\Delta_{\alpha,\mathbf{\Phi}}^k$, we have the following inequality
\begin{equation}\label{delta alpha phi}
\begin{split}
&\Delta_{\alpha,\mathbf{\Phi}}^k\geq\mathcal{L}(\alpha^{k}\!,\!\mathbf{\Phi}^{k}\!, \{\mathbf{\Phi}_n^k,\mathbf{\Lambda}_n^k, n\in\mathcal{T}\}) \\
&\hspace{2.3cm}-\!\mathcal{U}(\alpha^{k+1}\!,\mathbf{\Phi}^{k+1}\!,\!\{\mathbf{\Phi}_n^{k},\!\mathbf{\Lambda}_n^{k},n\in\mathcal{T}\}).
\end{split}
\end{equation}
 {
Letting $\alpha = \alpha^k$ and $\mathbf{\Phi}=\mathbf{\Phi}^{k}$ and plugging them into \eqref{diff U}, we can get
\[\mathcal{U}(\alpha^{k}\!,\!\mathbf{\Phi}^{k}\!, \{\mathbf{\Phi}_n^k,\mathbf{\Lambda}_n^k, n\!\in\!\mathcal{T}\})\!\leq\!\mathcal{L}(\alpha^{k}\!,\!\mathbf{\Phi}^{k}\!, \{\mathbf{\Phi}_n^k,\mathbf{\Lambda}_n^k, n\!\in\!\mathcal{T}\}).\]
Moreover, since $\mathcal{U}(\alpha^{k}\!,\!\mathbf{\Phi}^{k}\!, \{\mathbf{\Phi}_n^k,\mathbf{\Lambda}_n^k, n\in\mathcal{T}\})$ is up-bound of function $\mathcal{L}(\alpha^{k}\!,\!\mathbf{\Phi}^{k}\!, \{\mathbf{\Phi}_n^k,\mathbf{\Lambda}_n^k, n\in\mathcal{T}\})$, i.e., ,
\[\mathcal{U}(\alpha^{k}\!,\!\mathbf{\Phi}^{k}\!, \{\mathbf{\Phi}_n^k,\mathbf{\Lambda}_n^k, n\!\in\!\mathcal{T}\})\!\geq\!\mathcal{L}(\alpha^{k}\!,\!\mathbf{\Phi}^{k}\!, \{\mathbf{\Phi}_n^k,\mathbf{\Lambda}_n^k, n\!\in\!\mathcal{T}\}).\]
Combining the above two inequalities, we can conclude
\[\mathcal{U}(\alpha^{k}\!,\!\mathbf{\Phi}^{k}\!, \{\mathbf{\Phi}_n^k,\mathbf{\Lambda}_n^k, n\in\mathcal{T}\})=\mathcal{L}(\alpha^{k}\!,\!\mathbf{\Phi}^{k}\!, \{\mathbf{\Phi}_n^k,\mathbf{\Lambda}_n^k, n\in\mathcal{T}\}).\]
}
Then, inequality \eqref{delta alpha phi} can be further derived as follows
\begin{equation}\label{delta alpha phi 1}
\begin{split}
\Delta_{\alpha,\mathbf{\Phi}}^k&\geq\mathcal{U}(\alpha^{k}\!,\!\mathbf{\Phi}^{k}\!, \{\mathbf{\Phi}_n^k,\mathbf{\Lambda}_n^k, n\in\mathcal{T}\}) \\
&\hspace{1.8cm}-\!\mathcal{U}(\alpha^{k+1}\!,\mathbf{\Phi}^{k+1}\!,\!\{\mathbf{\Phi}_n^{k},\!\mathbf{\Lambda}_n^{k},n\in\mathcal{T}\}).
\end{split}
\end{equation}
 {
Since
 $\mathcal{U}(\alpha,\mathbf{\Phi}, \{\mathbf{\Phi}_n^k, \mathbf{\Lambda}_n^k, n\!\in\!\mathcal{T}\})$  is strongly quadratic convex function with respect to $\alpha$ and $\mathbf{\Phi}$ (see \eqref{upperbound U0}) \cite{boyd_04}, we have
 \begin{equation}\label{diff UU}
\begin{split}
&\mathcal{U}(\!\alpha^{k}\!,\!\mathbf{\Phi}^{k}\!, \{\mathbf{\Phi}_n^k,\mathbf{\Lambda}_n^k, n\!\in\!\mathcal{T}\})\\
&\hspace{1.7cm}-\!\mathcal{U}(\alpha^{k+1}\!,\mathbf{\Phi}^{k+1},\{\mathbf{\Phi}_n^{k},\mathbf{\Lambda}_n^{k},n\!\in\mathcal{T}\})\\
&\geq \!\langle\nabla_{\mathbf{\Phi}}\mathcal{U}(\alpha^{k+1}\!,\mathbf{\Phi}^{k+1}\!,\!\{\mathbf{\Phi}_n^{k},\!\mathbf{\Lambda}_n^{k},n\in\mathcal{T}\}), \mathbf{\Phi}^{k}\!-\!\mathbf{\Phi}^{k+1}\rangle \\
&+ \!\langle\nabla_{\alpha}\mathcal{U}(\alpha^{k+1}\!,\mathbf{\Phi}^{k+1}\!,\!\{\mathbf{\Phi}_n^{k},\!\mathbf{\Lambda}_n^{k},n\in\mathcal{T}\}), \alpha^{k}\!-\!\alpha^{k+1}\rangle \\
&+\frac{1}{2}(L \!+\! \sum\limits_{n\in \mathcal{T}}\rho_n)\|\mathbf{\Phi}^{k+1}\!-\!\mathbf{\Phi}^{k}\|_F^2 \!+\!\frac{L_{\alpha}}{2}|\alpha^{k+1}\!-\!\alpha^k|^2.
\end{split}
\end{equation}
Moreover, since
$\{\!\alpha^{k\!+\!1}\!,\!\mathbf{\Phi}^{k\!+\!1}\!\}\!\!\!=\!\!\!\!\underset{\alpha\!\in\!(0,\alpha_{\!\max}],\atop{{0\preceq\mathbf{\Phi}\prec2\pi}}} {\arg \min} \!\mathcal{U}\!\!\left(\!\alpha,\!\mathbf{\Phi}\!, \! \{\!\mathbf{\Phi}_n^k\!,\!\mathbf{\Lambda}_n^k\!,\!n\!\in\!\mathcal{T}\}\!\right)$,  we can get
 \begin{equation}\label{stationary two}
\begin{split}
  &\left\langle \! \nabla_{\mathbf{\Phi}}\mathcal{U}(\!\alpha^{k+1}\!,\mathbf{\Phi}^{k+1}\!,\{\mathbf{\Phi}_n^{k},\!\mathbf{\Lambda}_n^{k},n\!\in\!\mathcal{T}\}),\mathbf{\Phi}^{k}\!-\!\mathbf{\Phi}^{k+1}\right\rangle\!\geq\!0,\\
  &\langle\nabla_{\alpha}\mathcal{U}(\!\alpha^{k+1}\!,\mathbf{\Phi}^{k+1}\!,\{\mathbf{\Phi}_n^{k},\!\mathbf{\Lambda}_n^{k},n\!\in\!\mathcal{T}\!\}), \alpha^{k}-\alpha^{k+1}\rangle \geq 0.
\end{split}
\end{equation}
Plugging \eqref{diff UU} and \eqref{stationary two} into \eqref{delta alpha phi 1},
we can obtain
 \begin{equation}
\begin{split}
\Delta_{\alpha,\mathbf{\Phi}}^k\!\!
\geq\!\! \frac{1}{2}(L \!+\!\! \sum\limits_{n\in \mathcal{T}}\rho_n)\|\mathbf{\Phi}^{k+1}\!-\!\mathbf{\Phi}^{k}\|_F^2\! +\!\!\frac{L_{\alpha}}{2}|\alpha^{k+1}\!\!-\!\alpha^k|^2.
\end{split}
\end{equation}
}

For $\Delta_{\mathbf{\Phi}_n}^k$, it can be rewritten as
\begin{equation}\label{delta phi n 1}
\begin{split}
\Delta_{\mathbf{\Phi}_n}^k&\geq\mathcal{L}_n(\mathbf{\Phi}^{k+1}\!, \mathbf{\Phi}_n^k,\mathbf{\Lambda}_n^k)-\mathcal{U}_n(\mathbf{\Phi}^{k+1}\!,\mathbf{\Phi}_n^{k+1},\mathbf{\Lambda}_n^{k})\\
&=\mathcal{L}_n(\mathbf{\Phi}^{k+1}\!, \mathbf{\Phi}_n^k,\!\mathbf{\Lambda}_n^k)\!-\!\mathcal{U}_n(\mathbf{\Phi}^{k+1}\!,\!\mathbf{\Phi}_n^{k},\mathbf{\Lambda}_n^{k})\\
&\ \  +\mathcal{U}_n(\mathbf{\Phi}^{k+1}\!\!, \mathbf{\Phi}_n^k,\!\mathbf{\Lambda}_n^k)\!-\!\mathcal{U}_n(\mathbf{\Phi}^{k+1}\!,\!\mathbf{\Phi}_n^{k+1}\!,\!\mathbf{\Lambda}_n^{k}).
\end{split}
\end{equation}
 {
Plugging $\mathbf{\Phi}_n = \mathbf{\Phi}_n^{k}$ into \eqref{diff Un}, we can obtain
\[
  \begin{split}
  \mathcal{L}_n(\!\mathbf{\Phi}^{k+1}\!, \mathbf{\Phi}_n^k,\!\mathbf{\Lambda}_n^k)\!-\!\mathcal{U}_n(\!\mathbf{\Phi}^{k+1}\!,\!\mathbf{\Phi}_n^{k},\mathbf{\Lambda}_n^{k})
  \!\!\geq\!\! -\!2L_n\!\|{\mathbf{\Phi}_n^k}\!-\!{\mathbf{\Phi}^{k+1}}\|_{F}^2.
  \end{split}
\]
Because the functions $\mathcal{U}_n\left(\mathbf{\Phi}^{k+1},\mathbf{\Phi}_n,\mathbf{\Lambda}_n^{k}\right)$, $n\in\mathcal{T}$, are strongly convex with respect to $\mathbf{\Phi}_n$. We have
\[
  \begin{split}
  &\mathcal{U}_n\!(\!\mathbf{\Phi}^{k+1}\!\!, \mathbf{\Phi}_n^k,\mathbf{\Lambda}_n^k\!)\!\!-\!\!\mathcal{U}_n\!(\!\mathbf{\Phi}^{k+1}\!\!,\!\mathbf{\Phi}_n^{k+1}\!\!,\!\mathbf{\Lambda}_n^{k}\!) \!\!\geq\!\!\frac{\rho_n\!\!+\!\!L_n}{2}\|\mathbf{\Phi}_n^{k+1}\!\!\!-\!\!\mathbf{\Phi}_n^{k}\|_F^2,
  \end{split}
\]
Plugging the above two inequations into \eqref{delta phi n 1}, we have}
\begin{equation}
\begin{split}
\!\!\!\!\Delta_{\mathbf{\Phi}_n}^k&\!\!\geq\!-2L_n\|{\mathbf{\Phi}_n^k}\!-\!{\mathbf{\Phi}^{k+1}}\|_F^2\!+\!\frac{\rho_n\!+\!L_n}{2}\|\mathbf{\Phi}_n^{k+1}\!-\!\mathbf{\Phi}_n^{k}\|_F^2,\\
\vspace{-20pt}&\!\geq\!\!-4L_n\!\|{\mathbf{\Phi}_n^{k+1}}\!\!-\!{\mathbf{\Phi}^{k+1}}\|_F^2\!+\!\frac{\rho_n\!\!-\!7\!L_n}{2}\|\mathbf{\Phi}_n^{k+1}\!\!-\!\mathbf{\Phi}_n^{k}\|_F^2.\\
\end{split}
\end{equation}
Plugging \eqref{relax ADMM c} into the above inequality, we can get
\begin{equation}\label{delta phi n}
\!\Delta_{\mathbf{\Phi}_n}^k\!\geq\!\frac{-4L_n}{\rho_n^2}\!\|\mathbf{\Lambda}_n^{k+1}-\mathbf{\Lambda}_n^{k}\|_F^2\!+\!\frac{\rho_n\!-7\!L_n}{2}\|\mathbf{\Phi}_n^{k+1}\!-\!\mathbf{\Phi}_n^{k}\|_F^2.
\end{equation}
Furthermore, plugging \eqref{difference lambda} into \eqref{delta phi n}, we can see that $\Delta_{\mathbf{\Phi}_n}^k$ is lower bounded by
\begin{equation}\label{bound delta phi n}
\!\Delta_{\mathbf{\Phi}_n}^k\!\!\!\geq\!\!\frac{\rho_n^3\!\!-\!\!7\rho_n^2\!L_n\!\!-\! {32L_n^3}}{2\rho_n^2}\!\|\mathbf{\Phi}_n^{k+1}\!\!-\!\mathbf{\Phi}_n^{k}\|_{\!F}^2
\!-\!\frac{24L_n^3}{\rho_n^2}\|\mathbf{\Phi}^{k+1}\!-\!\mathbf{\Phi}^{k}\|_{\!F}^2.
\end{equation}

For $\Delta_{\mathbf{\Lambda}_n}^k$, through similar derivations, we have
\begin{equation}\label{delta lambda}
\Delta_{\mathbf{\Lambda}_n}^k {\geq}-\frac{2L_n^2}{\rho_n}\left(2\|{\mathbf{\Phi}_n^{k+1}}-{\mathbf{\Phi}_n^{k}}\|_F^2+3\|{\mathbf{\Phi}^{k+1}}-{\mathbf{\Phi}^{k}}\|_F^2\right).
\end{equation}

Plugging \eqref{delta alpha phi 1}, \eqref{bound delta phi n}, and \eqref{delta lambda} into \eqref{diff L}, we can obtain
\[
  \begin{split}
 \!\!\!\!\! &\mathcal{L}(\alpha^{k},\mathbf{\Phi}^k, \{\mathbf{\Phi}_n^k,\mathbf{\Lambda}_n^k,n\in\mathcal{T}\}) \\
  &\hspace{1.5cm}-\!\mathcal{L}(\alpha^{k+1},\!\mathbf{\Phi}^{k+1},\{\mathbf{\Phi}_n^{k+1},\mathbf{\Lambda}_n^{k+1},n\in\mathcal{T}\}) \\
\!\!\!\!\!\geq& \sum_{n\in \mathcal{T}}\frac{1}{2\rho_n^2}\left( \bar{c}_{n}\|{\mathbf{\Phi}_n^{k+1}}\!-\!{\mathbf{\Phi}_n^k}\|_F^2\!+\!\tilde{c}_{n} \|{\mathbf{\Phi}^{k+1}}\!-\!{\mathbf{\Phi}^{k}}\|_F^2 \right) \\
\!\!\!\!\!&+\frac{L}{2}\|{\mathbf{\Phi}^{k+1}}-{\mathbf{\Phi}^{k}}\|_F^2 +\frac{L_{\alpha}}{2}|\alpha^{k+1}-\alpha^k|^2,
  \end{split}
\]
where $\bar{c}_{n}\!=\!\rho_n^3\!-\!7\rho_n^2L_n\!-\! {8\rho_nL_n^2\!-\!32L_n^3}$ and $\tilde{c}_{n}\!=\! \rho_n^3\!-\!12\rho_nL_n^2\!-\!48L_n^3$. So, $\forall n \in \mathcal{T}$, if $\bar{c}_{n}>0$ and $\tilde{c}_{n}>0$, then the augmented Lagrangian function $\mathcal{L}(\cdot)$ {\it decreases sufficiently}. This completes the proof. $\hfill\blacksquare$
\end{proof}

\begin{lemma}\label{lemma lower bound}
  {if $\rho_n>5L_n$, augmented Lagrangian function  $\mathcal{L}(\alpha^{k+1},\mathbf{\Phi}^{k+1},\{\mathbf{\Phi}_n^{k+1},\mathbf{\Lambda}_n^{k+1}, n\in\mathcal{T}\})\geq 0$, $\forall k$}.
 \end{lemma}
\begin{proof}
 First, we consider  $\mathcal{L}_n({\mathbf{\Phi}}^{k+1},\mathbf{\Phi}_n^{k+1},\mathbf{\Lambda}_n^{k+1})$. Plugging \eqref{lambda update} into $\mathcal{L}_n({\mathbf{\Phi}}^{k+1},\mathbf{\Phi}_n^{k+1},\mathbf{\Lambda}_n^{k+1})$,  it can be equivalent to
\begin{equation}\label{lag bound}
\begin{split}
&\mathcal{L}_n\left(\mathbf{\Phi}^{k+1}, \mathbf{\Phi}_n^{k+1},\mathbf{\Lambda}_n^{k+1}\right) \\
=&  f_n({\mathbf{\Phi}_n^{k+1}})+(\frac{\rho_n}{2}-L_n)\|\mathbf{\Phi}_n^{k+1}-\mathbf{\Phi}^{k+1}\|_F^2 \\
&\hspace{2.5cm}+\langle\nabla  f_n({\mathbf{\Phi}}^{k+1}), \mathbf{\Phi}^{k+1}-\mathbf{\Phi}_n^{k+1}\rangle.
\end{split}
\end{equation}
Since $\nabla f_n(\mathbf{\Phi})$ is Lipschitz continuous,  {the last term in \eqref{lag bound} satisfies the following inequality}
\[
  \begin{split}
    &\langle\nabla  f_n({\mathbf{\Phi}}^{k+1}), \mathbf{\Phi}^{k+1}-\mathbf{\Phi}_n^{k+1}\rangle \\
    \geq& \langle\nabla f_n({\mathbf{\Phi}}_n^{k+\!1}), \mathbf{\Phi}^{k+1}-\mathbf{\Phi}_n^{k+1}\rangle -L_n\|\mathbf{\Phi}_n^{k+1}-\mathbf{\Phi}^{k+1}\|^2_F.
  \end{split}
\]
Plugging it into \eqref{lag bound}, we can get
\begin{equation}\label{lag bound 2}
\begin{split}
& \mathcal{L}_n\left(\mathbf{\Phi}^{k+1}, \mathbf{\Phi}_n^{k+1},\mathbf{\Lambda}_n^{k+1}\right) \\
\geq& f_n({\mathbf{\Phi}_n^{k+1}})+\langle\nabla f_n({\mathbf{\Phi}}_n^{k+\!1}), \mathbf{\Phi}^{k+1}-\mathbf{\Phi}_n^{k+1}\rangle \\
&\hspace{1.9cm} +(\frac{\rho_n}{2}-2L_n)\|\mathbf{\Phi}_n^{k+1}-\mathbf{\Phi}^{k+1}\|_F^2.
\end{split}
\end{equation}
Moreover, we can also exploit $\nabla f_n(\mathbf{\Phi}_n)$'s Lipschitz continuous property and get the following inequality
\[
  \begin{split}
  &f_n({\mathbf{\Phi}^{k+1}})\leq f_n({\mathbf{\Phi}_n^{k+1}})+\langle\nabla f_n({\mathbf{\Phi}}_n^{k+\!1}), \mathbf{\Phi}^{k+1}-\mathbf{\Phi}_n^{k+1}\rangle \\
  &\hspace{3.4cm}+ \frac{L_n}{2}\|\mathbf{\Phi}_n^{k+1}-\mathbf{\Phi}^{k+1}\|_F^2.
  \end{split}
\]
Plugging it into \eqref{lag bound 2}, we can obtain
\begin{equation}\label{lag bound 3}
\begin{split}
&\mathcal{L}_n\left(\mathbf{\Phi}^{k+1}, \mathbf{\Phi}_n^{k+1},\mathbf{\Lambda}_n^{k+1}\right)\\
&\hspace{1cm}\geq  f_n({\mathbf{\Phi}^{k+1}})+ \frac{\rho_n-5L_n}{2}\|\mathbf{\Phi}_n^{k+1}-\mathbf{\Phi}^{k+1}\|_F^2.
\end{split}
\end{equation}

Second, plugging \eqref{lag bound 3} into $\mathcal{L}(\alpha^{k+1},\mathbf{\Phi}^{k+1},\!\{\mathbf{\Phi}_n^{k+1},\!\mathbf{\Lambda}_n^{k+1}\!,\! n\!\in\!\mathcal{T}\}\!)$, we can get
\begin{equation}\label{upp func bound}
\begin{split}
&\mathcal{L}(\alpha^{k+1},\mathbf{\Phi}^{k+1},\{\mathbf{\Phi}_n^{k+1},\mathbf{\Lambda}_n^{k+1}, n\in\mathcal{T}\})\\
 \geq& h(\alpha^{k+1},\mathbf{\Phi}^{k+1}) \\
&\hspace{0.15cm}+\sum_{n\in \mathcal{T}}\left(\!f_n(\mathbf{\Phi}^{k+1})+\frac{\rho_n\!-\!5L_n}{2}\|\mathbf{\Phi}_n^{k+1}\!-\!\mathbf{\Phi}^{k+1}\|_F^2\right).
\end{split}
\end{equation}
 {Since $h(\alpha^{k+1},\mathbf{\Phi}^{k+1})\geq0$ and $f_n(\mathbf{\Phi}^{k+1})\geq0$, we can conclude that, if $\rho_n>5L_n$,  $\mathcal{L}(\alpha^{k+1},\mathbf{\Phi}^{k+1},\{\mathbf{\Phi}_n^{k+1},\mathbf{\Lambda}_n^{k+1}, n\in\mathcal{T}\})\geq 0$.} This completes the proof. $\hfill\blacksquare$
\end{proof}

\section{Proof of theorem 1}
First, we prove \eqref{convergence variables} in Theorem 1.

 {The pre-conditions that satisfying Lemma \ref{successive difference} and Lemma \ref{lemma lower bound} hold are
 \begin{equation}\label{ineq_3}
  \begin{split}
  \bar{c}_{n}\!&=\!\rho_n^3\!-\!7\rho_n^2L_n\!-\!8\rho_nL_n^2\!-\!32L_n^3>0, \\
  \tilde{c}_{n}\!&=\! \rho_n^3\!-\!12\rho_nL_n^2\!-\!48L_n^3>0, \\
  \rho_n&>5L_n.
  \end{split}
  \end{equation}
Letting $\beta=\frac{\rho_n}{L_n}$ and plugging $\beta$ into the above inequalities, we can obtain
 \[
  \begin{split}
  &\beta^3-7\beta^2-8\beta-32>0, \\
  &\beta^3-12\beta-48>0, \\
  &\beta>5.
  \end{split}
  \]
The roots of the cubic function in the left side of the first inequality can be determined through the famous Cardano formula \cite{Cardano}, which are $4.72, -2.36+2.15i, -2.36-2.15i$.} Similarly, we can obtain the roots of the cubic function in the second inequality, which are  $8.41,-0.70 + 1.82i,-0.70 - 1.82i$. Combining the above results with $\beta>5$, we can find that when $\beta\geq8.41$, all the inequalities in \eqref{ineq_3} hold simultaneously. To simplify the description, we choose $\rho_n\geq 9L_n$.
Therefore, Lemma \ref{successive difference} and Lemma \ref{lemma lower bound} are tenable when $\rho_n\geq 9L_n, \forall n \in \mathcal{T}$.

Since Lemma \ref{successive difference} holds, we sum both sides of the inequality \eqref{succ-diff phi} when $k=1,2,\dotsb,+\infty$ and obtain
\[
  \begin{split}
  &\mathcal{L}(\alpha^{1},\mathbf{\Phi}^1, \{\mathbf{\Phi}_n^1,\mathbf{\Lambda}_n^1,n\in\mathcal{T}\}) \\
  &\hspace{1.2cm}-\lim_{k\rightarrow+\infty}\!\mathcal{L}(\alpha^{k+1},\!\mathbf{\Phi}^{k+1},\{\mathbf{\Phi}_n^{k+1},\mathbf{\Lambda}_n^{k+1},n\in\mathcal{T}\}) \\
\geq& \sum_{k=1}^{+\infty}\sum_{n\in \mathcal{T}}\frac{1}{2\rho_n^2}\left( \bar{c}_{n}\|{\mathbf{\Phi}_n^{k+1}}\!-\!{\mathbf{\Phi}_n^k}\|_F^2\!+\!\tilde{c}_{n} \|{\mathbf{\Phi}^{k+1}}\!-\!{\mathbf{\Phi}^{k}}\|_F^2 \right) \\
&\ \ +\sum_{k=1}^{+\infty}\frac{L}{2}\|{\mathbf{\Phi}^{k+1}}-{\mathbf{\Phi}^{k}}\|_F^2 +\sum_{k=1}^{+\infty}\frac{L_{\alpha}}{2}|\alpha^{k+1}-\alpha^k|^2.
  \end{split}
\]
 {
Since $\displaystyle\lim_{k\rightarrow+\infty}\!\!\mathcal{L}(\alpha^{k\!+\!1},\!\mathbf{\Phi}^{k\!+\!1},\{\mathbf{\Phi}_n^{k\!+\!1},\mathbf{\Lambda}_n^{k\!+\!1},\!n\!\in\!\mathcal{T}\}) \!\geq\!0$, we can get the following inequality
\begin{equation*}\label{del L}
  \begin{split}
\hspace{-0.1cm}&\mathcal{L}(\alpha^{1},\mathbf{\Phi}^1, \{\mathbf{\Phi}_n^1,\mathbf{\Lambda}_n^1,n\in\mathcal{T}\})\\
  &\geq\! \sum_{k=1}^{+\infty}\sum_{n\in \mathcal{T}}\frac{1}{2\rho_n^2}\left( \bar{c}_{n}\|{\mathbf{\Phi}_n^{k+1}}\!-\!{\mathbf{\Phi}_n^k}\|_F^2\!+\!\tilde{c}_{n} \|{\mathbf{\Phi}^{k+1}}\!-\!{\mathbf{\Phi}^{k}}\|_F^2 \right) \\
&\hspace{0.2cm}\ \ +\sum_{k=1}^{+\infty}\frac{L}{2}\|{\mathbf{\Phi}^{k+1}}-{\mathbf{\Phi}^{k}}\|_F^2 +\sum_{k=1}^{+\infty}\frac{L_{\alpha}}{2}|\alpha^{k+1}-\alpha^k|^2.
  \end{split}
\end{equation*}
 Since $\bar{c}_n\!>\!0$, $\tilde{c}_n\!>\!0$, $L>0$, and $L_{\alpha}>0$, the above inequality indicates that summation of infinite positive terms is less than some constant.} Therefore, we can obtain \eqref{convergence limit PHI n}, \eqref{convergence limit PHI}, and \eqref{convergence limit alpha}.
\begin{equation}
\lim\limits_{k\rightarrow+\infty}\|\mathbf{\Phi}_n^{k+1}-\mathbf{\Phi}_n^{k}\|_F= 0, ~\forall~ n\in \mathcal{T}. \label{convergence limit PHI n}
\end{equation}
\begin{equation} \lim\limits_{k\rightarrow+\infty}\|\mathbf{\Phi}^{k+1}-\mathbf{\Phi}^{k}\|_F= 0. \label{convergence limit PHI}
\end{equation}
\begin{equation}
\lim\limits_{k\rightarrow+\infty}|\alpha^{k+1}-\alpha^{k}|= 0. \label{convergence limit alpha}
\end{equation}
Plugging \eqref{convergence limit PHI n}, \eqref{convergence limit PHI} into \eqref{difference lambda}'s right side, we can get
\begin{equation}\label{conv Lambda}
\lim\limits_{k\rightarrow+\infty}\|\mathbf{\Lambda}_n^{k+1}\!-\!\mathbf{\Lambda}_n^{k}\|_F\! = \!0.
\end{equation}
Combining \eqref{conv Lambda} and \eqref{relax ADMM c}, we further have
\begin{equation}\label{convergence diff}
\lim\limits_{k\rightarrow+\infty}\|\mathbf{\Phi}_n^{k+1}-\mathbf{\Phi}^{k+1}\|_F = 0.
\end{equation}

 { Since $\alpha\in [0, \alpha_{\max})$ and $0\preceq\mathbf{\Phi}\prec2\pi$, we can obtain the following convergence results from \eqref{convergence limit PHI} and \eqref{convergence limit alpha}.
\begin{subequations}\label{convergence_alpha_Phi}
    \begin{align}
    &\lim\limits_{k\rightarrow+\infty}\alpha^{k}\!=\!\alpha^{*}, \label{convergence_alpha_Phi a}\\ &\lim\limits_{k\rightarrow+\infty}\mathbf{\Phi}^{k}\!=\!\mathbf{\Phi}^{*}. \label{convergence_alpha_Phi b}
    \end{align}
\end{subequations}
Plugging \eqref{convergence_alpha_Phi b} into \eqref{convergence diff}, we can conclude
\begin{equation}\label{convergence_Phin}
\lim\limits_{k\rightarrow+\infty}\mathbf{\Phi}_n^{k}=\mathbf{\Phi}_n^{*}=\mathbf{\Phi}^{*}. \end{equation}
Plugging \eqref{convergence diff} into \eqref{lambda update}, we can derive
  \begin{equation}\label{g_lambda}
   \lim\limits_{k\rightarrow+\infty} \mathbf{\Lambda}_n^k = -\nabla f_n({\mathbf{\Phi}^k}).
   \end{equation}
Since $\|\nabla f_n({\mathbf{\Phi}})\!-\!\nabla f_n( \hat{\mathbf{\Phi}})\|_F\!\leq \! L_n\|{\mathbf{\Phi}}\!-\!\hat{\mathbf{\Phi}}\|_F, \ n\in \mathcal{T}$ and $0\preceq \mathbf{\Phi},\hat{\mathbf{\Phi}} \prec2\pi$, we can conclude that all the elements in $\nabla f_n({\mathbf{\Phi}})$ are bounded. From \eqref{g_lambda}, it indicates  that $\mathbf{\Lambda}_n^k$ is also bounded. Combining this result with \eqref{conv Lambda}, we can get
\begin{equation}\label{convergence lambda}
\lim\limits_{k\rightarrow+\infty}\mathbf{\Lambda}_n^{k}\!=\!\mathbf{\Lambda}_n^{*}, \forall n\in \mathcal{T}.
\end{equation}}

Second, we prove $(\alpha^*,\mathbf{\Phi}^{*})$ is a stationary point of problem \eqref{problem_without_constraint}.

 {
Since $\{\!\alpha^{k+1}\!,\mathbf{\Phi}^{k+1}\!\}\!\!=\!\!\!\!\underset{\alpha\in(0, \alpha_{\max}],\atop{{0\preceq\mathbf{\Phi}\prec2\pi}}} {\arg \min}\! \mathcal{U}\!\left(\alpha,\!\mathbf{\Phi}, \! \{\mathbf{\Phi}_n^k,\mathbf{\Lambda}_n^k,n\in\mathcal{T}\}\right)$  and  $\mathcal{U}\left(\alpha,\mathbf{\Phi}, \{\mathbf{\Phi}_n^k, \mathbf{\Lambda}_n^k, n\in\mathcal{T}\}\right)$ is convex quadratic function with respect to $\alpha$ and $\mathbf{\Phi}$, we have the following stationary conditions.
\[
\begin{split}
\hspace{-0.2cm}&\langle\nabla_{\!\alpha}\mathcal{U}\!\left(\!\alpha^{k\!+\!1}\!,\mathbf{\Phi}^{k\!+\!1}\!\!,\{\!\mathbf{\Phi}_n^k, \mathbf{\Lambda}_n^k,\!n\!\in\!\mathcal{T}\!\}\!\right)\!,\!\alpha\!-\!\alpha^{k\!+\!1}\rangle\!\!\geq\!\!0,\!\alpha\!\in\![0,\!\alpha_{\rm max}), \\
\hspace{-0.2cm}&\langle\nabla_{\!\mathbf{\Phi}}\mathcal{U}\!\left(\!\alpha^{k\!+\!1}\!,\mathbf{\Phi}^{k\!+\!1}\!, \{\!\mathbf{\Phi}_n^k, \! \mathbf{\Lambda}_n^k\!,\!n\!\in\!\mathcal{T}\!\}\!\right)\!,\!\mathbf{\Phi}\!-\!\mathbf{\Phi}^{k\!+\!1}\rangle\! \geq \!\!0,{0\!\preceq\!\mathbf{\Phi}\!\prec\!2\pi}.
\end{split}
\]
Plugging
\begin{equation*}
 \begin{split}
   & \nabla_{\alpha} \mathcal{U}(\alpha^{k+1},\mathbf{\Phi}^{k+1}, \{\mathbf{\Phi}_n^k, \mathbf{\Lambda}_n^k, n\!\in\!\mathcal{T}\})\\
    &\hspace{2.8cm}= \nabla_{\alpha} h(\alpha^k,{\mathbf{\Phi}}^{k})\!+ L_{\alpha}(\alpha^{k+1}\!-\!\alpha^{k}),\\
   &  \nabla_{\mathbf{\Phi}} \mathcal{U}(\alpha^{k+1},\mathbf{\Phi}^{k+1}, \{\mathbf{\Phi}_n^k, \mathbf{\Lambda}_n^k, n\!\in\!\mathcal{T}\})\!= \!\nabla_{\mathbf\Phi} h(\alpha^k, {\mathbf{\Phi}}^{k})\!\!\\
    &\hspace{1.8cm} +L({\mathbf{\Phi}^{k+1}}\!-\!{\mathbf{\Phi}}^{k})\! -\!\sum_{n\in \mathcal{T}}\left( \mathbf{\Lambda}_n^k\!+\!\rho_n({\mathbf{\Phi}}^{k+1}\!-\!{\mathbf{\Phi}_n^{k}})\right)
 \end{split}
\end{equation*}
into the above stationary conditions, we can get
\begin{equation}\label{stationary h_ori}
\begin{split}
&\big\langle\!\nabla_{\alpha}\!h\!\left(\!\alpha^k\!,\mathbf{\Phi}^{k}\right)\!\!+\!\!L_{\alpha}\!(\!\alpha^{k\!+\!1}\!\!\!-\!\!\alpha^k),\!\alpha\!\!-\!\alpha^{k\!+\!1}\!\big\rangle\!\geq\!0,\alpha\!\in\![0,\alpha_{\rm max}\!),\\
&\bigg\langle\!\nabla_{\mathbf{\Phi}}h\left(\!\alpha^{k},\mathbf{\Phi}^k\right)\!+\!L(\mathbf{\Phi}^{k\!+\!1}\!-\!\mathbf{\Phi}^k)\\
&-\!\!\!\sum_{n\in\mathcal{T}}\!\!\!\left( \rho_n(\mathbf{\Phi}^{k\!+\!1}\!\!-\!\mathbf{\Phi}_n^k)\!\!+\!\mathbf{\Lambda}_n^k\right)\!,\!\mathbf{\Phi}\!-\!\mathbf{\Phi}^{k\!+\!1}\!\bigg\rangle\!\!\geq\!0,{0\!\preceq\!\mathbf{\Phi}\!\prec\!2\pi}.
\end{split}
\end{equation}}

When $k\rightarrow+\infty$, plugging convergence results \eqref{convergence_alpha_Phi}, \eqref{convergence_Phin}, and \eqref{convergence lambda} into \eqref{stationary h_ori}, we can obtain
\begin{subequations}\label{stationary h}
\begin{align}
&\big\langle\nabla_{\alpha}h\left(\alpha^*,\mathbf{\Phi}^*\right),\alpha-\alpha^*\big\rangle\geq0, \ \alpha\in [0,\alpha_{\max}), \label{stationary h alpha}\\
&\bigg\langle\nabla_{\mathbf{\Phi}}h\!\left(\!\alpha^*,\mathbf{\Phi}^*\right)\!-\!\!\sum_{n\in\mathcal{T}}\mathbf{\Lambda}_n^*,\mathbf{\Phi}\!-\!\mathbf{\Phi}^*\bigg\rangle\geq0, {{0\!\preceq\!\mathbf{\Phi}\!\prec\!2\pi}.\label{stationary h phi}}
\end{align}
\end{subequations}
Since $\nabla f_n({\mathbf{\Phi}}^*)= - \mathbf{\Lambda}_n^*, \forall n\in\mathcal{T}$, then \eqref{stationary h phi} can be further derived as
\begin{equation}\label{stationary phi}
\bigg\langle\!\!\nabla_{\mathbf{\Phi}}h\left(\alpha^*,\mathbf{\Phi}^*\right)\!+\!\!\sum_{n\in\mathcal{T}}\!\nabla f_n({\mathbf{\Phi}}^*),\mathbf{\Phi}\!-\!\mathbf{\Phi}^*\!\!\bigg\rangle\!\!\geq\!0, {0\!\preceq\!\mathbf{\Phi}\!\prec\!2\pi}.
\end{equation}
Moreover, since $e\left(\alpha,\mathbf{X}(\mathbf{\Phi})\right)=h(\alpha,\mathbf{\Phi})$ and $P_c\left(\mathbf{X}(\mathbf{\Phi})\right)=\sum\limits_{n\in\mathcal{T}} f_n({\mathbf{\Phi}})$,
\eqref{stationary h alpha} and \eqref{stationary phi} can be rewritten as
\[
\begin{split}
&\langle\nabla_{\alpha} e\left(\alpha^*,\mathbf{X}(\mathbf{\Phi}^*)\right),{\alpha}-{\alpha^*}\rangle\geq0,\ \alpha\in (0,\alpha_{\max}], \\
&\langle\nabla_{\mathbf{\Phi}} e\left(\alpha^*,\mathbf{X}(\mathbf{\Phi}^*)\right)\!+\!\nabla  P_c\left(\mathbf{X}(\mathbf{\Phi}^*)\right),{\mathbf{\Phi}}\!-\!{\mathbf{\Phi}^{*}}\rangle\!\geq\!0, {0\!\preceq\!\mathbf{\Phi}\!\prec\!2\pi},
\end{split}
\]
which completes the proof. $\hfill\blacksquare$

\section{Proof of the convergence of the proposed algorithm with SBCD method}
 {
 Performing expectation on both sides of \eqref{succ-diff phi}, we can get the following inequality.
\begin{equation}\label{expectation sum diff}
\begin{split}
\!\!\!\!&\mathbb{E}\bigg[\mathcal{L}(\alpha^{k},\mathbf{\Phi}^k, \{\mathbf{\Phi}_n^k,\mathbf{\Lambda}_n^k, n\in\mathcal{T}\})\\
&\hspace{2cm}-\mathcal{L}(\alpha^{k\!+\!1},\mathbf{\Phi}^{k\!+\!1},\{\mathbf{\Phi}_n^{k\!+\!1},\mathbf{\Lambda}_n^{k\!+\!1}, n\in\mathcal{T}\})\bigg]\\
\!\!\!\!\!\geq& \sum_{n\in \mathcal{T}}\frac{p_{\rm min}}{2\rho_n^2}\left( \bar{c}_{n}\|{\mathbf{\Phi}_n^{k+1}}\!-\!{\mathbf{\Phi}_n^k}\|_F^2\!+\!\tilde{c}_{n} \|{\mathbf{\Phi}^{k+1}}\!-\!{\mathbf{\Phi}^{k}}\|_F^2 \right) \\
\!\!\!\!\!&+\frac{L}{2}\|{\mathbf{\Phi}^{k+1}}-{\mathbf{\Phi}^{k}}\|_F^2 +\frac{L_{\alpha}}{2}|\alpha^{k+1}-\alpha^k|^2,
\end{split}
\end{equation}
where $p_{\min}\geq0$ is the probability and $\bar{c}_{n}\!=\!\rho_n^3\!-\!7\rho_n^2L_n\!-\!8\rho_nL_n^2\!-\!32L_n^3$ and $\tilde{c}_{n}\!=\! \rho_n^3\!-\!12\rho_nL_n^2\!-\!48L_n^3$. We set $\rho_n\geq 9L_n,\forall n \in \mathcal{T}$ to guarantee $\bar{c}_{n}>0$ and $\tilde{c}_{n}>0$. Then, the augmented Lagrangian function {\it decreases sufficiently} in each consensus-ADMM-SBCD iteration.}

 {
Performing expectation on both sides of \eqref{upp func bound}, we can get the following inequality.
\begin{equation}\label{pro lower bound}
\begin{split}
\!\!\!\!&\mathbb{E}\bigg[\!\mathcal{L}(\alpha^{k+1}\!\!,\mathbf{\Phi}^{k+1}\!\!,\!\{\mathbf{\Phi}_n^{k+1}\!\!,\mathbf{\Lambda}_n^{k+1}\!,\!n\in\mathcal{T}\})\!\bigg]\!\!\geq \!\!h(\alpha^{k+1},\mathbf{\Phi}^{k+1}) \\
\!\!\!\!&\hspace{0.5cm}+\!\!\sum_{n\in \mathcal{T}}\!p_n\!\left(\!f_n(\mathbf{\Phi}^{k+1})\!+\!\frac{\rho_n\!-\!5L_n}{2}\|\mathbf{\Phi}_n^{k+1}\!-\!\mathbf{\Phi}^{k+1}\|_F^2\!\right),
\end{split}
\end{equation}
where $p_n>0$ is the probability. Since $h(\alpha^{k+1},\mathbf{\Phi}^{k+1})\geq0$ and $\rho_n\geq5L_n, \forall n \in \mathcal{T}$, from \eqref{pro lower bound} we can conclude
\begin{equation}\label{bound L}
\displaystyle\lim_{k\rightarrow+\infty}\mathbb{E}\left[\mathcal{L}(\alpha^{k+1},\!\mathbf{\Phi}^{k+1},\{\mathbf{\Phi}_n^{k+1},
\mathbf{\Lambda}_n^{k+1},n\in\!\mathcal{T}\})\right]\!\geq\!0.
\end{equation}}

 {Summing both sides of the inequality \eqref{expectation sum diff} for $k=1,2,\dotsb,+\infty$, we can obtain
 \[
  \begin{split}
\hspace{-0.1cm}&\mathbb{E}\bigg[\mathcal{L}(\alpha^{1},\mathbf{\Phi}^1, \{\mathbf{\Phi}_n^1,\mathbf{\Lambda}_n^1,n\in\mathcal{T}\})\!\bigg]\\
&\hspace{1cm} -\displaystyle\lim_{k\rightarrow+\infty}\mathbb{E}\left[\mathcal{L}(\alpha^{k\!+\!1},\!\mathbf{\Phi}^{k\!+\!1},\{\mathbf{\Phi}_n^{k\!+\!1},
\mathbf{\Lambda}_n^{k\!+\!1},n\in\!\mathcal{T}\})\right] \\
 &\geq\! \sum_{k=1}^{+\infty}\sum_{n\in \mathcal{T}}\frac{p_{\rm min}}{2\rho_n^2}\left( \bar{c}_{n}\|{\mathbf{\Phi}_n^{k+1}}\!-\!{\mathbf{\Phi}_n^k}\|_F^2\!+\!\tilde{c}_{n} \|{\mathbf{\Phi}^{k+1}}\!-\!{\mathbf{\Phi}^{k}}\|_F^2 \right) \\
&+\sum_{k=1}^{+\infty}\frac{L}{2}\|{\mathbf{\Phi}^{k+1}}-{\mathbf{\Phi}^{k}}\|_F^2 +\sum_{k=1}^{+\infty}\frac{L_{\alpha}}{2}|\alpha^{k+1}-\alpha^k|^2.
  \end{split}
 \]
Plugging \eqref{bound L} into above inequality, it can be simplified as
 \begin{equation*}
  \begin{split}
\hspace{-0.1cm}&\mathbb{E}\bigg[\mathcal{L}(\alpha^{1},\mathbf{\Phi}^1, \{\mathbf{\Phi}_n^1,\mathbf{\Lambda}_n^1,n\in\mathcal{T}\})\!\bigg]\\
 &\geq\! \sum_{k=1}^{+\infty}\sum_{n\in \mathcal{T}}\frac{p_{\rm min}}{2\rho_n^2}\left( \bar{c}_{n}\|{\mathbf{\Phi}_n^{k+1}}\!-\!{\mathbf{\Phi}_n^k}\|_F^2\!+\!\tilde{c}_{n} \|{\mathbf{\Phi}^{k+1}}\!-\!{\mathbf{\Phi}^{k}}\|_F^2 \right) \\
&+\sum_{k=1}^{+\infty}\frac{L}{2}\|{\mathbf{\Phi}^{k+1}}-{\mathbf{\Phi}^{k}}\|_F^2 +\sum_{k=1}^{+\infty}\frac{L_{\alpha}}{2}|\alpha^{k+1}-\alpha^k|^2.
  \end{split}
  \end{equation*}}

 {Since $\bar{c}_n\!>\!0$, $\tilde{c}_n\!>\!0$, $L>0$, and $L_{\alpha}>0$, the above inequality indicates that summation of infinite positive terms is less than some constant. Then, we can conclude the following results
\begin{equation}
\begin{split}
&\lim\limits_{k\rightarrow+\infty}\|\mathbf{\Phi}_n^{k+1}-\mathbf{\Phi}_n^{k}\|_F= 0, ~\forall~ n\in \mathcal{T},\\
& \lim\limits_{k\rightarrow+\infty}\|\mathbf{\Phi}^{k+1}-\mathbf{\Phi}^{k}\|_F= 0,\\
&\lim\limits_{k\rightarrow+\infty}|\alpha^{k+1}-\alpha^{k}|= 0.
\end{split}
\end{equation}
Since \eqref{relax ADMM c} and \eqref{difference lambda} hold, we can further get
\begin{equation}\label{convergence Lambda}
\begin{split}
&\lim\limits_{k\rightarrow+\infty}\|\mathbf{\Lambda}_n^{k+1}\!-\!\mathbf{\Lambda}_n^{k}\|_F\!\! = \!\!0, \\ &\lim\limits_{k\rightarrow+\infty}\|\mathbf{\Phi}_n^{k+1}\!-\!\mathbf{\Phi}^{k+1}\|_F \!= \!0,\forall n\in \mathcal{T}.
\end{split}
\end{equation}
Through similar discussions \eqref{convergence_alpha_Phi} to \eqref{convergence lambda}, we can obtain the following convergence results
\begin{equation}
\begin{split}
&\lim\limits_{k\rightarrow+\infty}\alpha^{k}=\alpha^{*}, \ \ \ \lim\limits_{k\rightarrow+\infty}\mathbf{\Phi}^{k}=\mathbf{\Phi}^{*},\\
&\lim\limits_{k\rightarrow+\infty}\mathbf{\Phi}_n^{k}=\mathbf{\Phi}_n^{*}, \ \
\lim\limits_{k\rightarrow+\infty}\mathbf{\Lambda}_n^{k}=\mathbf{\Lambda}_n^{*},\\
& \ \ \mathbf{\Phi}^{*}=\mathbf{\Phi}_n^{*},\ \ \forall~n\in \mathcal{T},
\end{split}
\end{equation}
which concludes the proof.}
 $\hfill\blacksquare$

\ifCLASSOPTIONcaptionsoff
  \newpage
\fi

\begin{spacing}{1.0}

\begin{IEEEbiography}[{\includegraphics[width=1in,height=1.3in,clip,keepaspectratio]{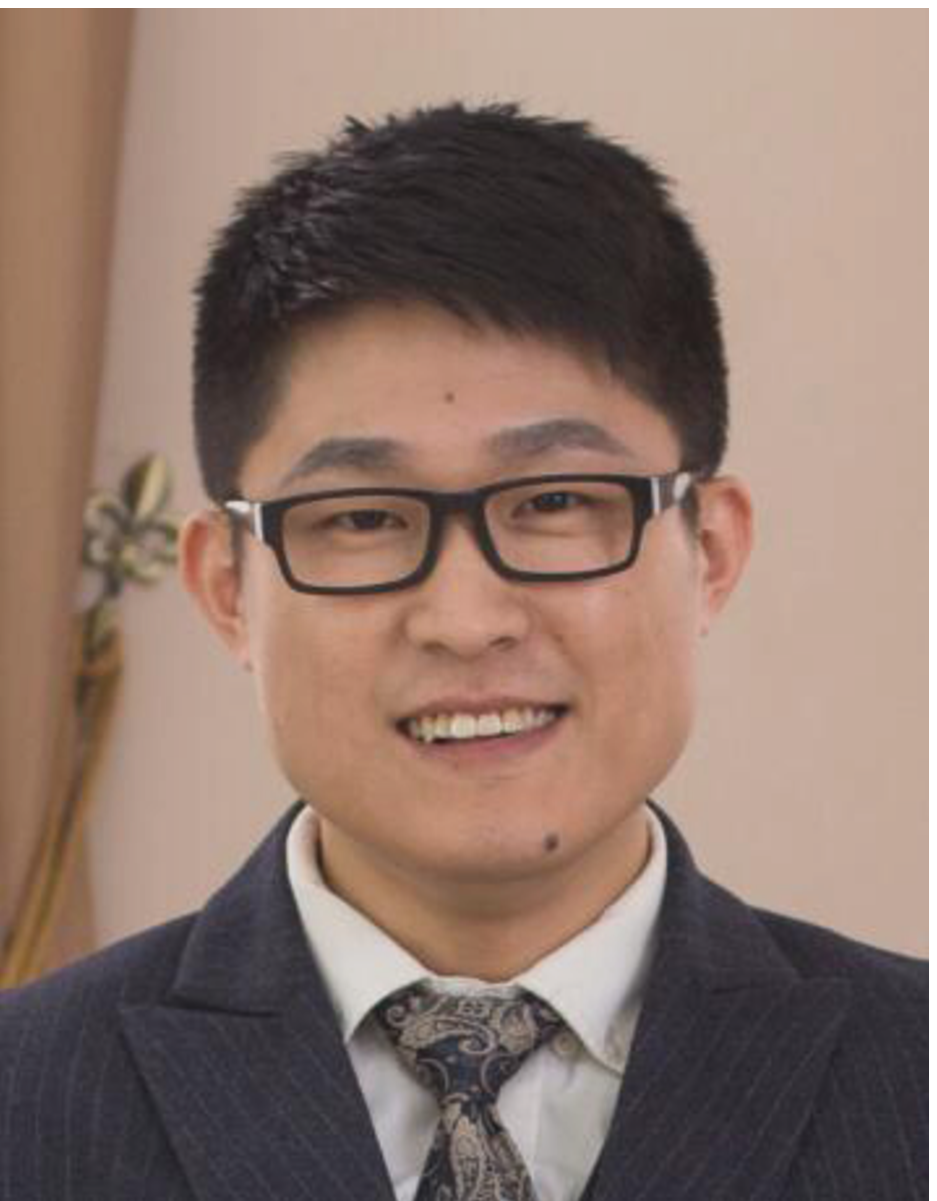}}]
{\bf Jiangtao Wang} received the B.E. degree in electronic information engineering from Shandong Normal University,
Jinan, China, in 2011. He is currently working toward the Ph.D. degree in communication and information systems in Xidian University.
His research interests are convex optimization and efficient algorithms with applications in sequence set design, MIMO radar waveform design and communication signal processing.
\end{IEEEbiography}

\begin{IEEEbiography}[{\includegraphics[width=1in,height=1.3in,clip,keepaspectratio]{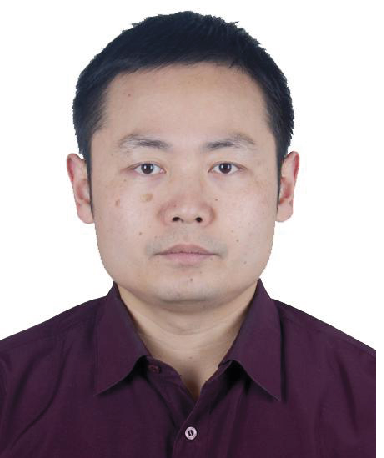}}]
{\bf Yongchao Wang}
received the B.E. degree in communication engineering, M.E. and Ph.D. degrees in information and communication engineering from
Xidian University, Xi¡¯an, China, in 1998, 2004, and 2006 respectively. From September 2008 to January 2010, he was a one-year visiting scholar and then
a post-doctoral in ECE department of University of Minnesota, USA. From March 2016 to March 2017, he was a visiting scholar in Purdue University, USA.
From 2012, he owns a full processor position of ISN key state Lab. in Xidian University.
His research interests mainly lie in the areas of signal processing for communications, machine learning algorithms and their applications. To date, he has published
more than 30 peer-reviewed papers as the first author or corresponding author and issued about 20 national patents. Moreover, he is a recipient of several
awards, such as prize in progress of science and technology from Ministry of Education of the People¡¯s Republic of China, Shaanxi government, and
Xidian University. His several research works have been applied to several real telecommunication systems.
\end{IEEEbiography}

\end{spacing}


\begin{thebibliography}{99}

\bibitem{Li_09}
J. Li and P. Stoica, \emph{MIMO Radar Signal Processing.} Hoboken, NJ, USA: Wiley, 2009.
\bibitem{Li_07}
J. Li and P. Stoica, ``MIMO radar with colocated antennas,''\emph{ IEEE Signal Process. Mag.}, vol. 24, no. 5, pp. 106-114, Sep. 2007.
\bibitem{Haimovich_08}
A. M. Haimovich, R. S. Blum, and L. J. Cimini, ``MIMO radar with widely separated antennas,'' \emph{IEEE Signal Process. Mag.}, vol. 25, pp. 116-129, Jan. 2008.

\bibitem{He_10}
Q. He, R. S. Blum, H. Godrich, and A. M. Haimovich, ``Target velocity estimation and antenna placement for MIMO radar with widely separated
antennas,'' \emph{IEEE Journal of Selected Topics in Signal Processing}, vol. 4, no. 1, pp. 79-100, Feb. 2010.
\bibitem{Lehmann_06}
H. Lehmann, A. Haimovich, R. Blum, and L. Cimini, ``High resolution capabilities of MIMO radar,'' in \emph{Proc. 40th Asilomar Conf. Signals, Syst., Comput.}, Pacific Grove, CA, pp. 25-30, Nov. 2006.
\bibitem{Stoica_07_radar}
P. Stoica, J. Li, and X. Zhu, ``MIMO radar with co-located antenna: Review of some recent work,'' \emph{IEEE Signal Process. Mag.}, vol. 24, pp. 106-114, Sep. 2007.
\bibitem{Chen_08}
C.-Y. Chen and P. P. Vaidyanathan, ``MIMO radar space-time adaptive processing using prolate spheroidal wave functions,'' \emph{IEEE Trans. Signal Process.}, vol. 56, no. 2, pp. 623-635, Feb. 2008.
\bibitem{Stoica_07_probing}
P. Stoica, J. Li, and Y. Xie, ``On probing signal design for MIMO radar,'' \emph{IEEE Trans. Signal Process.}, vol. 55, no. 8, pp. 4151-4160, Aug. 2007.
\bibitem{Li_08}
J. Li, P. Stoica, and X.-Y. Zheng, ``Signal synthesis and receiver design for MIMO radar imaging,'' \emph{IEEE Trans. Signal Process.}, vol. 56, no. 8, pp. 3959-3968, Aug. 2008.

\bibitem{Wang_12}
Y.-C. Wang, X. Wang, H. Liu, and Z.-Q. Luo, ``On the design of constant modulus probing signals for MIMO radar,'' \emph{IEEE Trans. Signal Process.}, vol. 60, no. 8, pp. 4432-4438, Aug. 2012.
\bibitem{Hua_13}
G. Hua, and S. S. Abeysekera, ``MIMO radar transmit beampattern design with ripple and transition band control,'' \emph{IEEE Trans. Signal Process.}, vol. 61, no. 11, pp. 2963-2974. Jun. 2013.

\bibitem{Ahmed_2014}
S. Ahmed and M. Alouini, ``MIMO radar waveform covariance matrix for high SINR and low side-lobe levels,'' \emph{IEEE Trans. Signal Process.}, vol. 62, no. 8, pp. 2056-2065, Apr. 2014.
\bibitem{Zhang_15}
X. Zhang, Z. He, L. Rayman-Bacchus, and J. Yan, ``MIMO radar transmit beampattern matching design,'' \emph{IEEE Trans. Signal Process.}, vol. 63, no. 8 pp. 2049-2056, Apr. 2015.
\bibitem{Lipor_14}
J. Lipor, S. Ahmed and M. S. Alouini, ``Fourier-based transmit beampattern design using MIMO radar,'' \emph{IEEE Trans. Signal Process.}, vol. 62, no. 9, pp. 2226-2235, Mar. 2014.
\bibitem{Bouchoucha_15}
T. Bouchoucha, S. Ahmed, T. Y. Naffouri, and M. S. Alouini, ``Closed form solution to directly design face waveforms for beampatterns using
planar array,'' in \emph{Proc. 40th IEEE Int. Conf. on Acoust., Speech, Signal Process. (ICASSP)}, Brisbane, Australia, Apr. 2015, pp. 2359-2363.
\bibitem{Bouchoucha_17}
T. Bouchoucha, S. Ahmed and M. S. Alouini, ``DFT-based closed form covariance matrix and direct waveforms design for MIMO radar
to achieve desired beampatterns,'' \emph{IEEE Trans. Signal Process.}, vol. 62, no. 9, pp. 2104-2113, Jan. 2017.
\bibitem{Aubry_16}
A. Aubry, A. De Maio, and Y. Huang, ``MIMO radar beampattern design via PSL/ISL optimization.'' \emph{IEEE Trans. Signal Process.}, vol. 64, no. 15, pp. 3955-3967, Aug. 2016.
\bibitem{Aldayel_17_conf}
 O. Aldayel, V. Monga and M. Rangaswamy, ``Transmit MIMO beampattern design under constant modulus and spectral interference constraints,'' in \emph{Proc. IEEE Int. Radar Conf.}, Seattle, Washington, USA, May 2017, pp. 1131-1136.
\bibitem{Aldayel_17}
O. Aldayel, V. Monga and M. Rangaswamy, ``Tractable transmit MIMO beampattern design under a constant modulus constraint,'' \emph{IEEE Trans. Signal Process.}, vol. 65, no. 10, pp. 2588-2599, May 2017.

\bibitem{Cheng_17}
Z. Cheng, Z. He, S. Zhang and J. Li, ``Constant modulus waveform design for MIMO radar transmit beampattern,'' \emph{IEEE Trans. Signal Process.}, vol. 65, no. 18, pp. 4912-4923. Sep. 2017.
\bibitem{Cheng_18}
 Z. Cheng, Y. Lu, Z. He, J. Li, and  X. Luo, ``Joint optimization of covariance matrix and antenna position for MIMO radar transmit beampattern matching design,'' in \emph{Proc. IEEE Int. Radar Conf.}, Oklahoma City, OK, USA, Apr. 2018, pp. 1073-1077.
 \bibitem{Zhao_18}
Z. Zhao and D. P. Palomar, ``MIMO transmit beampattern matching under waveform constraints,'' in \emph{Proc. 43rd IEEE Int. Conf. on Acoust., Speech, Signal Process. (ICASSP)}, Calgary, Alberta, Canada, Apr. 2018, pp. 3281-3285.

\bibitem{Yu_18}
X. Yu, G. Cui, T. Zhang, and L. Kong, ``Constrained Transmit Beampattern Design for Colocated MIMO Radar,'' Signal Process., vol. 144, pp. 145-154, Mar. 2018.
\bibitem{Yu_19}
X. Yu, G. Cui, J. Yang, and L. Kong, ``Wideband MIMO Radar Beampattern Shaping with Space-Frequency Nulling,'' Signal Process., vol. 160, pp. 80-87, Jul. 2019.

 \bibitem{Wen_18}
 F. Wen, J. Liang, and J. Li, ``Constant modulus MIMO radar waveform design with minimum peak sidelobe transmit beampattern," \emph{IEEE Trans. Signal Process.} vol. 66 no. 16, pp. 4207-4222, Aug. 2018.

\bibitem{Liang_15}
J. Liang, H. C. So, C. S. Leung, J. Li, and A. Farina, ``Waveform design with unit modulus and spectral shape constraints via lagrange programming neural network,'' \emph{IEEE J. Sel. Topics Signal Process.}, vol. 9, no. 8, pp. 1377-1386, Dec. 2015.
\bibitem{Liang_16}
J. Liang, H. C. So, J. Li, and A. Farina, ``Unimodular sequence design based on alternating direction method of multipliers,'' \emph{IEEE Trans. Signal Process.}, vol. 64, no. 20, pp. 5367-5381, Oct. 2016.
\bibitem{Cheng_14}
Z. Cheng, Z. He, Li, M. Fang, J. Li and J. Xie, `` Spectrally compatible waveform design for MIMO radar transmit beampattern with PAR and similarity constraints,''  in \emph{Proc. 39th IEEE Int. Conf. Acoust., Speech Signal Process.}, Florence, Italy, May 2014, pp. 5312-5316.
\bibitem{Liang_18}
J. Liang, X. Zhang, H. C. So and D. Zhou, ``Sparse array beampattern synthesis via alternating direction method of multipliers,'' \emph{IEEE Trans. Signal Process.} vol. 66, no. 5, pp. 2333-2345. May 2018.
\bibitem{Li_17}
J. Li, G. Liao, J. Xu and Y. Huang, ``A flexible transmit beam pattern design approach based on waveform covariance matrix,'' in \emph{Proc. IEEE Int. Radar Conf.}, Seattle, Washington, USA, May 2017, pp. 1308-1312.
\bibitem{ADMM_con}
 Y. Wang and J. Wang, ``Constant modulus probing waveform design for mimo radar via ADMM algorithm,'' in \emph{Proc. 43rd IEEE Int. Conf. on Acoust., Speech, Signal Process. (ICASSP)}, Calgary, Alberta, Canada, Apr. 2018, pp. 3305-3309.
\bibitem{Hong_16}
M. Hong, Z. Luo, and M. Razaviyayn, ``Convergence analysis of alternating direction method of multipliers for a family of nonconvex
problems,'' \emph{SIAM Journal on Optimization}, vol. 26, no. 1, pp. 337-364, Jan. 2016.
\bibitem{Kerahroodi_17}
M. A. Kerahroodi, A. Aubry, A. De Maio, M. M. Naghsh, and M. Modarres-Hashemi, ``A coordinate-descent framework to design low PSL/ISL sequences,'' \emph{IEEE Trans. Signal Process.}, vol.65, no. 22, pp. 5942-5956, Nov. 2017.
\bibitem{Zhao_17}
L. Zhao, J. Song, P. Babu, and D. Palomar, ``A unified framework for low autocorrelation sequence design via majorization-minimization,''\emph{IEEE Trans. Signal Process.}, vol. 65, no. 2, pp. 438-453, Jan. 2017.
\bibitem{Aubry_18}
A. Aubry, A. D. Maio, A. Zappone, M. Razaviyayn, and Z.-Q. Luo, ``A new sequential optimization procedure and its applications to resource allocation for wireless systems,'' \emph{IEEE Trans. Signal Process.} vol. 66, no. 24, pp. 6518-6533, Dec. 2018.
\bibitem{Bertsekas_99}
D. P. Bertsekas, Nonlinear programming, Athena scientific, second edition, pp. 667-668, 1999.
\bibitem{Tseng_01}
Paul. Tseng, ``Convergence of a block coordinate descent method for nondifferentiable minimization.'' \emph{Journal of Optimization Theory and Applications,} vol. 109, no. 3, pp. 475-494, Jun. 2001.
\bibitem{Nesterov_83}
Y. Nesterov, ``A method of solving a convex programming problem with convergence rate o(1/k2),'' \emph{Soviet Mathematics Doklady}, vol. 27, no. 2, pp. 372-376, 1983.
\bibitem{Nesterov_04}
Y. Nesterov, \emph{Introductory Lectures on Convex Optimization}. New York, NY: Kluwer Academic Press, 2004.
\bibitem{Meng_acc}
X. Meng and H. Chen, ``Accelerating Nesterov's method for strongly convex functions with Lipschitz gradient,'' [Online]. Available: https://arxiv.org/abs/1109.6058.
\bibitem{Goldstein_14}
T. Goldstein, B. O'Donoghue, S. Setzer and R. Baraniuk, ``Fast alternating direction optimization methods,'' \emph{SIAM Journal on Imaging Sciences}, vol. 7, no. 3, pp. 1588-1623, Aug. 2014.
\bibitem{He_15}
B. He and X. Yuan, ``On non-ergodic convergence rate of Douglas-Rachford alternating direction method of multipliers,'' \emph{Numerische Mathematik}, vol. 130, no. 3, pp. 567-577, Jul. 2015.
\bibitem{Hong_17}
M. Hong and Z.-Q. Luo, ``On the linear convergence of the alternating direction method of multipliers,'' \emph{Math. Program.}, vol. 162, no. 1, pp. 165-199, Mar. 2017.

\bibitem{botao_global_convergence}
Y. Wang, W. Yin, and J. Zeng, ``Global convergence of ADMM in nonconvex nonsmooth optimization,'' [Online]. Available: https://arxiv.org/abs/1511.06324.
\bibitem{boyd_04}
S. Boyd and L. Vandenberghe, \emph{Convex Optimization}. Cambridge, U.K.: Cambridge Univ. Press, 2004.
\bibitem{Cardano}
R. Witu{\l}a and D. S{\l}ota, ``Cardano's formula, square roots, Chebyshev polynomials and radicals,''. \emph{Journal of Mathematical Analysis and Applications}, vol. 2, no. 363, pp. 639-647, 2010.

\end{thebibliography}
\end{document}